\documentclass[conference,10pt]{IEEEtran}

\usepackage{cite}

\ifCLASSINFOpdf
   \usepackage[pdftex]{graphicx}
\else
\fi

\usepackage{amssymb,amsmath,amsthm}
\usepackage{bm}
\usepackage{stmaryrd}
\usepackage{tikz}
\usetikzlibrary{arrows,automata,calc,shapes.arrows,backgrounds,intersections,positioning,snakes}
\tikzset{
    every picture/.style={>=stealth,auto,node distance=2cm,}
}
\tikzstyle{pstate}=[rectangle,
    inner sep=0pt,
    minimum size=0.5cm,
    draw= black,
]
\tikzstyle{cstate}=[state,
    inner sep=0pt,
    minimum size=0.5cm,
    draw= black,
]
\usepackage[
  textsize=small]{todonotes}

\usepackage{macros}

\usepackage{thm-restate}
\usepackage[capitalise,noabbrev,nameinlink]{cleveref}

\newtheorem{theorem}{Theorem}
\newtheorem{lemma}[theorem]{Lemma}

\theoremstyle{definition}
\newtheorem{definition}[theorem]{Definition}
\newtheorem{example}[theorem]{Example}
\newtheorem{remark}[theorem]{Remark}
\newtheorem{claim}[theorem]{Claim}
\crefname{lemma}{Lemma}{Lemmas}
\crefname{definition}{Definition}{Definitions}
\theoremstyle{plain}

\crefname{proposition}{Proposition}{Propositions}
\theoremstyle{remark}
{\itshape}{\rmfamily}
\crefname{fact}{Fact}{Facts}
\crefname{claim}{Claim}{Claims}

\renewcommand{\pi}{\tau}

\begin{document}
\title{MDPs with Energy-Parity Objectives}

\author{
    \IEEEauthorblockN{Richard Mayr\IEEEauthorrefmark{1},
    Sven Schewe\IEEEauthorrefmark{2},
    Patrick Totzke\IEEEauthorrefmark{1},
Dominik Wojtczak\IEEEauthorrefmark{2}}
\IEEEauthorblockA{\IEEEauthorrefmark{1}University of Edinburgh, UK}
\IEEEauthorblockA{\IEEEauthorrefmark{2}University of Liverpool, UK}
}

\maketitle

\begin{abstract}
Energy-parity objectives combine $\omega$-regular with quantitative objectives of reward MDPs.
The controller needs to avoid to run out of energy while satisfying a parity objective.

We refute the common belief that, if an energy-parity objective holds almost-surely,
then this can be realised by some finite memory strategy.
We provide a surprisingly simple counterexample that only uses coB\"uchi conditions.

We introduce the new class of bounded (energy) storage objectives that,
when combined with parity objectives, preserve the finite memory property.
Based on these, we show that almost-sure and limit-sure energy-parity
objectives, as well as almost-sure and limit-sure storage parity objectives, are in $\NP\cap\coNP$ and can be solved in pseudo-polynomial time for energy-parity MDPs.

\end{abstract}

\newcommand{\headerblob}[1]{\smallskip\noindent\textbf{#1.}}

\section{Introduction}
\headerblob{Context}
Markov decision processes (MDPs) are a standard model for dynamic systems that
exhibit both stochastic and controlled behaviour \cite{Puterman:book}.
Such a process starts in an initial state and makes a sequence of transitions between states.
Depending on the type of the current state, either the controller gets to
choose an enabled transition (or a distribution over transitions), or the next
transition is chosen randomly according to a predefined distribution.
By fixing a strategy for the controller, one obtains a Markov chain.
The goal of the controller is to optimise the (expected) value of
some objective function on runs of such an induced Markov chain.

\headerblob{Our Focus and Motivation}
In this paper we study MDPs with a finite number of states, where
numeric rewards (which can be negative) are assigned 
to transitions. We consider 
quantitative objectives, e.g.~the total expected reward or
the limit-average expected reward \cite{Puterman:book,CD2011}.
Note that the total reward is not bounded \emph{a priori}.
We also consider $\omega$-regular objectives that can be expressed by parity conditions on 
the sequence of visited states (subsuming many simpler objectives like 
B\"uchi and coB\"uchi).

When reasoning about controllers for mechanical and electrical systems,
one may need to consider quantitative objectives such as the remaining stored energy
of the system (which must not fall below zero, or else the system fails),
and, at the same time, parity objectives that 
describe the correct behaviour based on the temporal specification.
Thus one needs to study the combined \emph{energy-parity} objective.

\headerblob{Status Quo}
Previous work in \cite{CD2011} (Sec.~3) considered the decidability and complexity of the
question whether the energy-parity objective can be satisfied
\emph{almost-surely}, i.e.~whether there exists a strategy (or: a
controller) that satisfies the objective with probability~$1$.
They first show that in the restricted case of energy-B\"uchi objectives,
finite memory optimal strategies exist,
and that almost-sure satisfiability is in $\NP\cap\coNP$
and can be solved in pseudo-polynomial time.

They then describe a direct reduction from almost-sure energy-parity
to almost-sure energy-B\"uchi. 
This reduction claimed that the winning
strategy could be chosen among a certain subclass of strategies,
that we call \emph{colour-committing}. Such a strategy eventually commits to 
a particular winning even colour, where this colour must be seen infinitely often
almost-surely and no smaller colour must ever been seen after committing.
 
However, this reduction from almost-sure energy-parity
to almost-sure energy-B\"uchi in \cite{CD2011} (Sec.~3)
contains a subtle error (which also appears in the survey
in \cite{chatterjee2011games} (Theorem 4)).
In fact, we show that 
strategies for almost-sure energy-parity may require infinite memory.

\smallskip
\noindent{\bf Our contributions} can be summarised as follows.

\smallskip\textit{1)}
We provide a simple counterexample that shows that, even 
for almost-sure energy-coB\"uchi objectives, the winning strategy 
requires infinite memory
and cannot be chosen among the colour-committing strategies.

\smallskip\textit{2)}
We introduce an energy \emph{storage} objective, which requires that the
energy objective 
is met using a finite energy store. The size of the store can be fixed by the controller, but it cannot be changed.
We argue that the almost-sure winning sets for energy-B\"uchi and
storage-B\"uchi objectives coincide.
Moreover, we show that the reduction in \cite{CD2011} actually works for storage-parity instead of for energy-parity conditions.
I.e.\ \cite{CD2011} shows that almost-sure storage parity objectives require
just finite memory, are in $\NP\cap\coNP$, and can be solved in pseudo-polynomial time.

\smallskip\textit{3)}
We develop a solution for the original almost-sure 
energy-parity objective. It requires a more involved argument and
infinite-memory strategies that are obtained by composing three
other strategies.
We show that almost-sure energy-parity objectives are in $\NP\cap\coNP$ and can be solved in pseudo-polynomial time.
      
\smallskip\textit{4)}
We then study the \emph{limit-sure problem}.
Here one asks whether, for every $\epsilon >0$, there exists a strategy that 
satisfies the objective with probability $\ge 1-\epsilon$.
This is a question about the existence of a family of $\epsilon$-optimal strategies,
not about a single strategy as in the almost-sure problem.
The limit-sure problem is equivalent to the question whether the \emph{value}
of a given state and initial energy level (w.r.t.\ the objective) is $1$.
For the \emph{storage-parity} objective, the limit-sure condition
coincides with the almost-sure condition,
and thus the complexity results from \cite{CD2011} apply.
In contrast, for the \emph{energy-parity} objective
the limit-sure condition does \emph{not} coincide with the almost-sure
condition.
While almost-sure energy-parity implies limit-sure energy-parity, we give examples
that show that the reverse implication does \emph{not} hold.
We develop an algorithm to decide the limit-sure energy-parity objective
and show that the problem 
is in $\NP \cap \coNP$ and can be solved in pseudo-polynomial time.
Moreover, each member in the family of $\epsilon$-optimal strategies
that witnesses the limit-sure energy-parity condition can be chosen
as a finite-memory strategy (unlike winning strategies for
almost-sure energy-parity that may require infinite memory).

\headerblob{Related work} Energy games were introduced in \cite{chakrabarti2003resource} to reason about systems with multiple components and bounded resources. Energy objectives were later also considered in the context of timed systems \cite{BFLMS2008}, synthesis of robust systems \cite{bloem2009better}, and gambling \cite{BergerKSV08} (where they translate to ``not going bankrupt'').
The first analysis of a combined qualitative--quantitative objective was done in \cite{chatterjee2005mean} for mean-payoff parity games.
Almost-sure winning in energy-parity MDPs was considered in \cite{CD2011} (cf.~\cite{chatterjee2011games} as a survey).
However, it was shown in \cite{CD2011} that almost-sure winning in energy-parity MDPs
is at least as hard as two player energy-parity games \cite{CD2012}. 
A recent paper \cite{BKN:ATVA2016} considers a different combined objective:
maximising the expected mean-payoff while satisfying the energy objective.
The proof of Lemma 3 in \cite{BKN:ATVA2016} uses a reduction to the
(generally incorrect) result on energy-parity MDPs of \cite{CD2011},
but their result still holds because it only uses the correct
part about energy-B\"uchi MDPs.

Closely related to energy MDPs and games are one-counter MDPs and games, where the counter value can be seen as the current energy level. One-counter MDPs and games with a termination objective (i.e.\ reaching counter value $0$) were studied in \cite{BBEKW10} and \cite{BBE10}, respectively.

\headerblob{Outline of the Paper}
The following section introduces the necessary notations. \Cref{sec:claims} discusses
combined energy-parity objectives and formally states our results.
In \cref{sec:bug} we explain the error in the construction of \cite{CD2011} (Sec.~3),
define the bounded energy storage condition and derive the results on
combined storage-parity objectives.
\Cref{sec:as-EP} discusses the almost-sure problem for energy-parity conditions
and provides our new proof of their decidability.
The limit-sure problem for energy-parity is discussed in \cref{sec:limit,sec:ls-EP}.
\Cref{sec:complexity} discusses
lower bounds and the relation between almost/limit-sure parity MDPs and mean-payoff games. 
Due to the space constraints some details had to be omitted and can be found in the full version \cite{DBLP:journals/corr/MayrSTW17}.

\section{Notations}\label{sec:preliminaries}
A probability distribution over a set $X$ is a function $f:X\to[0,1]$
such that $\sum_{x\in X} f(x) = 1$. We write $\Dist{X}$ for the set of
distributions over $X$.

\headerblob{Markov Chains}
A \emph{Markov chain} is an edge-labeled, directed graph $\sys{C}\eqdef (V,E,\prob)$,
where the elements of $V$ are called \emph{states},
such that the labelling $\prob:E\to[0,1]$ provides a probability distribution
over the set of outgoing transitions of any state $s\in V$.
A \emph{path} is a finite or infinite sequence $\rho\eqdef s_1s_2\ldots$ of
states such that $(s_i,s_{i+1})\in E$ holds for all indices $i$; an infinite
path is called a \emph{run}.
We use $w\in V^*$ to denote a finite path.
We write $s\step{x}t$ instead of $(s,t)\in E \land \prob(s,t)=x$
and omit superscripts whenever clear from the context.
We write $\Runs[\sys{C}]{w}$ for the cone set $wV^\omega$, i.e., 
the set of runs with finite prefix $w\in V^*$,
and assign to it the probability space
$(\Runs[\sys{C}]{w},\mathcal{F}_{w}^\sys{C},\Prob[\sys{C}]{w})$,
where $\mathcal{F}_{w}^\sys{C}$ is the $\sigma$-algebra generated by
all cone sets $\Runs[\sys{C}]{wx}\subseteq \Runs[\sys{C}]{w}$,
for $x=x_1x_2\dots x_l\in V^*$.
The probability measure
$\Prob[\sys{C}]{w}:\mathcal{F}^{\sys{C}}_w\to[0,1]$ is defined as
$\Prob[\sys{C}]{w}(\Runs{wx}) \eqdef \Pi_{i=1}^{l-1} \lambda(x_i,x_{i+1})$
for cone sets.
By Carath\'eodory's extension theorem~\cite{billingsley-1995-probability},
this defines a unique probability measure on all measurable subsets of runs.

\headerblob{Markov Decision Processes}
A \emph{Markov Decision Process (MDP)} is a sinkless directed graph
$\sys{M}\eqdef(\VC,\VP, E, \prob)$,
where $V$ is a set of states, partitioned as $V\eqdef \VC\uplus\VP$ into
\emph{controlled} ($\VC$) and \emph{probabilistic} states ($\VP$).
The set of \emph{edges} is $E\subseteq V\x V$ and $\prob:\VP\to\Dist{E}$
assigns each probabilistic state a probability distribution over its outgoing edges.

A \emph{strategy} is a function $\sigma:V^*\VC\to\Dist{E}$ that assigns each
word $ws\in V^*\VC$ a probability distribution over the outgoing edges of $s$,
that is $\sigma(ws)(e)>0$ implies $e=(s,t)\in E$ for some $t\in V$.
A strategy is called \emph{memoryless} if $\sigma(xs)=\sigma(ys)$ for all $x,y\in V^*$ and $s\in \VC$,
and \emph{deterministic} if if $\sigma(w)$ is Dirac for all $w\in V^*\VC$.
Each strategy induces a Markov chain $\sys{M}(\sigma)$ with states $V^*$
and where $ws\step{x}wst$ if $(s,t)\in E$ and either
$s\in\VP \land x=\prob(s,t)$ or $s\in\VC\land x=\sigma(ws,t)$.
We write $\Runs[\sys{M}]{w}$ for the set of runs in $\sys{M}$ (with prefix
$w$), consisting of all runs in $\Runs[\sys{M}(\sigma)]{w}$ for some strategy
$\sigma$, and $\Runs[\sys{M}]{}$ for the set of all such paths.

\renewcommand{\O}{\mathsf{Obj}}
\headerblob{Objective Functions}
An \emph{objective} is a subset $\O\subseteq \Runs[\sys{M}]{}$.
We write $\overline{\O}\eqdef\Runs[\sys{M}]{}\setminus\O$ for its complement.
It is satisfied
\emph{surely} if there is a strategy $\sigma$ such that $\Runs[\sys{M}(\sigma)]{}\subseteq \O$,
\emph{almost-surely} if there exists a $\sigma$ such that $\Prob[\sys{M}(\sigma)]{}(\O) = 1$ and
\emph{limit-surely} if $\sup_\sigma\Prob[\sys{M}(\sigma)]{}(\O) = 1$.
In other words, the limit-sure condition asks that there exists some infinite sequence $\sigma_1,\sigma_2,\ldots$
of strategies such that $\lim_{n\to\infty} \Prob[\sys{M}(\sigma_n)]{}(\O) = 1$.
We call a strategy \emph{$\eps$-optimal} 
if $\Prob[\sys{M}(\sigma)]{}(\O) \ge 1-\eps$.
Relative to a given MDP $\sys{M}$ and some finite path $w$, we define the \emph{value} of $\O$ as
$\Val[\sys{M}]{w}{\O} \eqdef \sup_\sigma \Prob[\sys{M}(\sigma)]{w}(\O)$.
We use the following objectives, defined by conditions on individual runs.

\smallskip
A \emph{reachability condition} is defined by a set of target states $T \subseteq V$. 
A run $s_0s_1\ldots$ satisfies the reachability condition iff
there exists an $i \in \N$ s.t.\ $s_i \in T$.
We write $\eventually T \subseteq \Runs{}$ for the set of runs that satisfy 
the reachability condition.

\smallskip
A \emph{parity condition} is given by a function $\parity:V\to\N$,
that assigns a priority (non-negative integer) to each state.
A run $\rho \in \Runs{}$ satisfies the parity condition if
the minimal priority that appears infinitely often
on the run is even.
The \emph{parity objective} is the subset $\Parity[]{} \subseteq
\Runs{}$ of runs that satisfy the parity condition.

\smallskip
\emph{Energy conditions} are given by a function $\cost{}:E\to\Z$,
that assigns a \emph{cost} value to each edge.
For a given initial energy value $k\in\N$, a run $s_0s_1\ldots$ satisfies the $k$-energy condition
if, for every finite prefix, the energy level $k+\sum_{i=0}^l\cost(s_i,s_{i+1})$ stays greater or equal to $0$.
Let $\EN{k} \subseteq \Runs{}$ denote the $k$-energy objective,
consisting of those runs that satisfy the $k$-energy condition.

\smallskip
\emph{Mean-payoff conditions} are defined w.r.t.\ the same cost function
$\cost{}:E\to\Z$ as the energy conditions. 
A run $s_0s_1\ldots$ satisfies the \emph{positive mean-payoff condition} iff
$\liminf_{n\rightarrow\infty}\frac{1}{n}\sum_{i=0}^{n-1}\cost(s_i,s_{i+1}) > 0$.
We write $\PosMP\subseteq \Runs{}$ for the positive mean-payoff objective,
consisting of those runs that satisfy the positive mean-payoff condition.

\section{Parity Conditions under Energy Constraints}
\label{sec:claims}
We study the combination of energy and parity objectives for finite MDPs.
That is, given a MDP and both cost and parity functions,
we consider objectives of the form $\EN{k}\cap\Parity{}$
for integers $k\in \N$.
We are interested in identifying those control states and values $k\in\N$
for which the combined $k$-energy-parity objective 
is satisfied almost-surely and limit-surely, respectively.

\begin{example} 
    \label{ex:lval}
    Consider a controlled state $s$
    that can go left or right with cost $0$,
    or stay with cost $1$.
    The probabilistic state on the left
    increases or decreases energy with equal chance, whereas
    the probabilistic state on the right has a positive energy updrift.
    State $s$ has priority $1$, all other states have priority $0$.
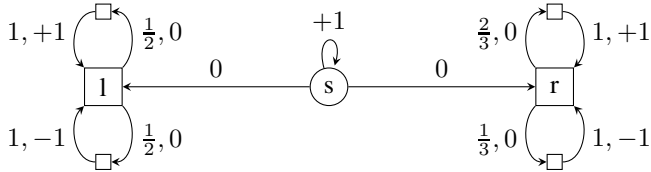
\begin{center}
  \begin{tikzpicture}[node distance=1cm and 3cm, on grid]
    \tikzstyle{acstate}=[pstate,minimum size=0.2cm]
    \node[cstate] (Y) {s};
    \node[pstate, right=of Y] (Z) {r};
    \node[pstate, left=of Y] (X) {l};
    \node[acstate, above=of X] (lla) {};
    \node[acstate, above=of Z] (rra) {};
    \node[acstate, below=of X] (llb) {};
    \node[acstate, below=of Z] (rrb) {};
    \draw[->] (Y) edge node[above] {$0$} (Z);
    \draw[->] (Y) edge node[above] {$0$} (X);
    \draw[->,in=west, out=south west] (Z) edge node[left]{${\frac{1}{3}},0$} (rrb);
    \draw[->,out=east, in=south east] (rrb) edge node[right] {$1,-1$}(Z);
    \draw[->,out=north west, in=west] (Z) edge node[left]{${\frac{2}{3}},0$} (rra);
    \draw[->,in=north east, out=east] (rra) edge node[right]{$1,+1$} (Z);
    \draw[->,in=east, out=north east] (X) edge node[right]{${\frac{1}{2}},0$} (lla);
    \draw[->,out=west, in=north west] (lla) edge node[left]{$1,+1$} (X);
    \draw[->,in=east, out=south east] (X) edge node[right] {${\frac{1}{2}},0$} (llb);
    \draw[->,out=west, in=south west] (llb) edge node[left] {$1,-1$} (X);
    \draw[->,loop above] (Y) edge node[] {$+1$} (Y);
\end{tikzpicture}
\end{center}
    From states other than $s$ there is only one strategy.
    It holds that $\Val{l}{\Parity{}}=1$ but
    $\Val{l}{\EN{k}}=0$ for any $k\in\N$ and so $\Val{l}{\EN{k}\cap\Parity{}}=0$.
    For state $r$ we have that $\Val{r}{\EN{k}\cap\Parity{}}=\Val{r}{\EN{k}}=1-(1/2)^k$,
    due to the positive drift.
    For all $k\in\N$ the state $s$ does not satisfy the $k$-energy-parity objective
    almost-surely but limit-surely: $\Val{s}{\EN{k}\cap \Parity{}}=1$
    (by going ever higher and then right).
\end{example}

Notice that these energy-parity objectives are trivially monotone in the parameter $k$ because $\EN{k}\subseteq\EN{k+1}$ holds for all $k\in\N$.
Consequently, for every fixed 
state $p$,
if there exists some $k\in N$
such that the $k$-energy-parity objective holds almost-surely (resp.\ limit-surely),
then there is a minimal such value $k$.
By \emph{solving} the almost-sure/limit-sure problems for these monotone objectives
we mean to compute these minimal sufficient values for all initial states.

\medskip
We now state our two main technical results.
We fix a finite MDP $\sys{M}\eqdef(\VC,\VP, E, \prob)$,
a parity function $\parity:V\to\N$ with maximal colour $d\in\N$
and a cost-function $\cost:E\to\Z$ with maximal absolute value $W \eqdef \max_{e\in E} |\cost(e)|$.
Let $|\lambda|$ and $|\cost{}|$ be the size of the transition table $\lambda$ and the cost function $\cost{}$, written as tables with valuations in binary.
We use $\widetilde {\mathcal O}(f(n))$ as a shorthand for ${\mathcal O}(f(n)\log^k f(n))$ for some constant $k$.

\begin{theorem}\label{thm:correction}
    \label{thm:as-energy-parity}
    (1) Almost-sure optimal strategies for $k$-energy-parity objectives may require infinite memory.
    (2) The almost-sure problem for $k$-energy-parity objectives is in $\NP\cap\coNP$
    and can be solved in pseudo-polynomial time
    $\widetilde {\mathcal O}(d\cdot|V|^{4.5} \cdot (|\lambda| + |\cost{}|)^2 +
    \card{E}\cdot d\cdot \card{V}^5\cdot W)$.
\end{theorem}

\begin{theorem}
    \label{thm:ls-energy-parity}
    \label{thm:main}
    (1) The limit-sure problem for $k$-energy-parity objectives is in $\NP\cap\coNP$
    and can be solved in pseudo-polynomial time
    $\widetilde {\mathcal O}(d\cdot|V|^{4.5} \cdot (|\lambda| + |\cost{}|)^2 +
    \card{E}\cdot d\cdot \card{V}^5\cdot W)$.
    (2) 
    If the $k$-energy-parity objective holds limit-surely
    then, for each $\eps>0$, there exists a \emph{finite memory} $\eps$-optimal strategy.
\end{theorem}

\begin{remark}
The claimed algorithms are \emph{pseudo polynomial} in the sense that they depend (linearly)
on the value $W$. If the cost-deltas are $-1,0,$ or $1$ only, and not arbitrary binary encoded numbers,
this provides a polynomial time algorithm.
\end{remark}

Part (2) of \cref{thm:as-energy-parity} was already claimed in \cite{CD2011}, Theorem~1.
However, the proof there relies on a particular finiteness assumption that is not true in general.
In the next section we discuss this subtle error and describe the class of
(bounded) \emph{storage} objectives,
for which this assumption holds and the original proof goes through.
Our new proof of \cref{thm:as-energy-parity} is presented in \cref{sec:as-EP}.

The proof of \cref{thm:ls-energy-parity} is deferred to \cref{sec:limit,sec:ls-EP}.
It is based on a reduction 
to checking almost-sure satisfiability of storage-parity objectives,
which can be done in pseudo polynomial time (cf.\ Theorem \ref{thm:storage}). 
We first establish in \cref{sec:limit}
that certain \emph{limit values} are computable for each state.
In \cref{sec:extensions} we then provide the actual reduction, which is based
on precomputing these limit values and produces an MDP which is only linearly
larger and has no new priorities.

\section{Energy Storage Constraints}
\label{sec:storage}
\label{sec:bug}
The argument of \cite{CD2011} to show computability of almost-sure energy-parity objectives
relies on the claim that
the controller, if it has a winning strategy, can eventually commit to visiting an even colour infinitely often
and \emph{never} visiting smaller colours.
We show that this claim already fails for coB\"uchi conditions (i.e.\ for MDPs that only use colours $1$ and $2$).
We then identify a stronger kind of energy condition---the storage energy condition we introduce below---that satisfies the above claim and for which the original proof of \cite{CD2011} goes
through.

Let us call a strategy \emph{colour-committing}
if, for some colour $2i$, almost all runs eventually visit a position
such that \emph{almost all} possible continuations visit 
colour $2i$ infinitely often and \emph{no continuation} (as this is a safety constraint) visits a colour smaller than $2i$.

\begin{claim}
    \label{claim}
    If there exists some strategy that
    almost-surely satisfies $\EN{k}\cap\Parity$
    then there is also a colour-committing strategy that does.
\end{claim}

\begin{proof}[Proof (that \cref{claim} is false)]
Consider the following example,
where the controller owns states $A,B,C$
and tries to avoid visiting state $B$
infinitely often while maintaining the energy condition.
This can be expressed as an energy-parity condition where
$
\parity(A)=
\parity(C)=
\parity(D)= 2
$ and $\parity(B)=1$.

\begin{center}
\begin{tikzpicture}[node distance=1.5cm and 3cm, on grid]
    \tikzstyle{acstate}=[cstate,minimum size=0.2cm,]
    \node[cstate] (A) {$A$};
    \node[cstate, below=of A] (B) {$B$};
    \node[cstate, right=of B] (C) {$C$};
    \node[pstate, right=of A] (D) {$D$};
    
    \draw[->] (A) edge[bend right] node[swap] {$1$} (D);
    \draw[->] (D) edge[bend right] node[swap] {$\frac{2}{3}$, $0$} (A);
    \draw[->] (D) edge[bend left] node {$\frac{1}{3}$, $0$} (C);
    \draw[->] (C) edge[bend left] node {$-1$} (D);
    \draw[->] (C) edge node {$0$} (B);
    \draw[->] (B) edge node {$0$} (A);
\end{tikzpicture}
\end{center}

First notice that all states almost-surely satisfy the $0$-energy-coB\"uchi condition $\EN{0}\cap\Parity$.
One winning strategy is to chooses the edge $C\step{-1}D$ over $C\step{0} B$, unless the energy level is $0$, in which case $C\step{0} B$ is favoured over $C\step{-1}D$.
This strategy is not colour-committing but clearly energy safe:
the only decreasing step is avoided if the energy level is $0$.

To see why this strategy also almost-surely satisfies the
parity (in this case coB\"uchi) objective,
first observe that it guarantees a positive updrift:
from state $D$ with positive energy level, the play returns to $D$ in two steps with expected energy gain $+1/3$,
from state $D$ with energy level $0$, the play returns to $D$ in either two or three steps, in both cases with energy gain $+1$.
The chance to visit state $C$ with energy level $0$ when starting at state $D$ with energy level $k\in\N$ is $(1/2)^{k+1}$.
This is the same likelihood with which state $B$ is eventually visited.
However, every time state $B$ is visited, the system restarts from state $D$ with energy level $1$.
Therefore, the chance of re-visiting $B$ from $B$ is merely $1/4$.
More generally, the chance of seeing state $B$ at least $n$ further times is $(1/4)^{n}$.
The chance of visiting $B$ infinitely often is therefore $\lim_{n \to \infty} (1/4)^n = 0$.
This strategy thus satisfies the parity---in this case coB\"uchi---objective almost-surely.
Consequently, the combined $0$-energy-parity objective is almost-surely met from all states.

\smallskip
To contradict \cref{claim}, we contradict the existence of an initial state and
a colour-committing strategy that almost-surely satisfies the $0$-energy-parity objective.
By definition, such a strategy will, on almost all runs, eventually avoid state $B$ completely.

As every run will surely visit state $D$ infinitely often,
we can w.l.o.g.\ pick a finite possible prefix $s_1s_2\ldots s_j$ (i.e.\ a prefix that can occur with a likelihood $\delta > 0$) of a run that ends in state $s_j=D$ and assume that none (or only a $0$ set, but these two conditions coincide for safety objectives) of its continuations visits state $B$ again.
Let $l\eqdef\sum_{i=1}^j\cost(s_i,s_{i+1})$ denote the sum of rewards collected on this prefix. 
Note that there is a $(1/3)^{l+1}>0$ chance that some continuation alternates between states $D$ and $C$ for $l+1$ times and thus violates the $l$-energy condition.
Consequently, the chance of violating the $0$-energy parity
condition from the initial state is at least $\delta\cdot(1/2)^{l+1}>0$.
\end{proof}

Notice that every finite memory winning strategy for the $\Parity$ objective must also be colour-committing.
The system above therefore also proves part (1) of \cref{thm:as-energy-parity}, that infinite memory
is required for $k$-energy-parity objectives.

\medskip
In the rest of this section we consider a stronger
kind of energy condition, for which \cref{claim} does hold
and the original proof of \cite{CD2011} goes through.
The requirement is that the strategy achieves the energy condition without being able to store an infinite amount of energy.
Instead, it has a finite energy store, say $s$, and cannot store more energy than the size of this storage.
Thus, when a transition would lead to an energy level $s' > s$, then it would result in an available energy of $s$.
These are typical behaviours of realistic energy stores, 
e.g.\ a rechargeable battery or a storage lake. 
An alternative view (and a consequence) is that the representation of the
system becomes finite-state once the bound $s$ is fixed, and only finite
memory is needed to remember the current energy level.

For the definition of a storage objective, we keep the infinite storage capacity, but instead require that no subsequence loses more than $s$ energy units.
The definitions are interchangeable, and we chose this one in order not to change the transitions of the system.

\begin{definition}
    \label{def:storage}
    For a finite MDP 
    with associated cost function, 
    a run $s_0s_1\ldots$
    satisfies the \emph{$s$-storage condition}
    if, for every infix $s_ls_{l+1}\ldots s_u$,
    it holds that $s+\sum_{i=l}^{u+1}\cost(s_i,s_{i+1})\ge 0$.
    Let $\ES{k,s} \subseteq \Runs{}$ denote the $k$-energy $s$-storage objective,
    consisting of those runs that satisfy both the $k$-energy
    and the $s$-storage condition.
\end{definition}

\begin{example}
\label{example:ks_tradeoff}
The two parameters can sometimes be traded against each other, as shown in the following example.
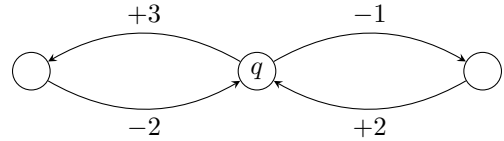
\begin{center}
\begin{tikzpicture}[node distance=1.5cm and 3cm, on grid]
    \node[cstate] (M) {$q$};
    \node[cstate, left=of M] (L) {};
    \node[cstate, right=of M] (R) {};
    
    \draw[->] (M) edge[bend right] node[swap] {$+3$}   (L);
    \draw[->] (L) edge[bend right] node[swap] {$-2$} (M);
    \draw[->] (M) edge[bend left]  node {$-1$} (R);
    \draw[->] (R) edge[bend left]  node {$+2$} (M);
\end{tikzpicture}
\end{center}
From state $q$ in the middle, one can win with an initial energy level $0$ by
always going left, provided that one has an energy store of size at least
$2$. With an energy store of size $1$, however, going left is not an option, as one would not be able to return from the state on the left. But with an initial energy level of $1$, one can follow the strategy to always turn to the right.
So the $\ES{0,2}$ and $\ES{1,1}$ objectives hold almost-surely but the $\ES{0,1}$ objective does not.
\end{example}

We sometimes want to leave the size of the energy store open.
For this purpose, we define $\ES{k}$ as the objective that says ``there is an $s$, such that $\ES{k,s}$ holds'' and $\mathsf{ST}$ for ``there is an $s$ such that $\ES{s,s}$ holds''.
Note that this is not a path property; we rather require that the $s$ is fixed globally.
In order to meet an $\ES{k}$ property almost-surely, there must be a strategy $\sigma$ and an $s\in \N$ such that almost all runs  satisfy $\ES{k,s}$: $\exists \sigma,s$ s.t.\ $\Prob[\sys{M}(\sigma)]{}(\ES{k,s}) = 1$.
Likewise, for limit-sure satisfaction of $\mathsf{ST}$, we require $\exists
s\ \forall \varepsilon >0\ \exists \sigma$ s.t.\  $\Prob[\sys{M}(\sigma)]{}(\ES{s,s}) \ge 1-\eps$.

We now look at combined storage-parity and storage-mean-payoff objectives.

\begin{theorem}[Storage-Parity]
    \label{thm:storage}
    \label{thm:as-storage-parity}
    For finite MDPs with states $V$, edges $E$ and associated cost and parity functions,
    with maximal absolute cost $W$ and maximal colour $d\in\N$:
    \begin{itemize}
        \item[1)] The almost-sure problem for storage-parity objectives 
    is in $\NP\cap\coNP$, and there is an algorithm to solve it in
    $\mathcal{O}(\card{E}\cdot d\cdot \card{V}^4\cdot W)$ deterministic time.
\item[2)] Memory of size $\mathcal{O}(\card{V}\cdot W)$ is sufficient for almost-sure winning strategies.
    This also bounds the minimal values $k,s\in\N$ such that $\ES{k,s}\cap\Parity$ holds almost-surely.
    \end{itemize}
\end{theorem}

The proof is provided by Chatterjee and Doyen \cite{CD2011}:
they first show the claim for energy-B\"uchi objectives $\EN{k}\cap\Parity$
(where $d=1$)
by reduction to two-player energy-B\"uchi \emph{games}
(\!\!\cite{CD2011}, Lemma~2).
Therefore, almost-sure winning strategies come from first-cycle games and operate in a bounded
energy radius. As a result, almost-sure satisfiability for energy-B\"uchi and storage-B\"uchi
coincide.
They then (\!\cite{CD2011}, Lemma~3) provide a reduction for general parity conditions
to the B\"uchi case, assuming \cref{claim}. Although this fails for energy-parity
objectives, as we have shown above, almost-sure winning strategies for storage-parity
can be assumed to be finite memory and therefore colour committing.
The construction of \cite{CD2011} then goes through without alteration.
The complexity bound follows from improvements for energy parity games \cite{CD2012}.

\begin{theorem}[Storage-Mean-payoff]
    \label{thm:as-ES-PosMP}
    \label{thm:as-storage-PosMP}
    \label{lem:bailout-to-parity}
        For finite MDPs with combined storage and positive mean-payoff objectives:
    \begin{itemize}
        \item[1)] The almost-sure problem is in $\NP\cap\coNP$
            and can be solved in
            $\mathcal{O}(\card{E}\cdot d\cdot \card{V}^4\cdot W)$ deterministic time.
        \item[2)] Memory of size $\mathcal{O}(\card{V}\cdot W)$ is sufficient for almost-sure winning strategies.
            This also bounds the minimal value $k,s\in\N$ such that $\ES{k,s}\cap\PosMP$ holds almost-surely.
    \end{itemize}
\end{theorem}
\begin{proof}
    We show that, for every MDP $\sys{M}$ with associated $\cost$ function,
    there is a linearly larger system $\sys{M'}$ with associated $\cost'$ and $\parity$ function
    ---where the parity function is B\"uchi, i.e.\ has image $\{0,1\}$---that, for every $k\in\N$,
    $\PosMP\cap\ES{k}$ holds almost-surely in $\sys{M}$ iff
    $\Parity\cap\ES{k}$ holds almost-surely in $\sys{M}'$.

    For every state $q$ of $\sys{M}$,
    the new system $\sys{M'}$ contains two new states, $q'$ and $q''$,
    edges $(q,q')$ and  $(q,q'')$ with costs $0$ and $-1$, respectively.
    Each original edge $(q,r)$ is replaced by two edges, $(q',r)$ and $(q'',r)$.
    All original states become controlled, and the primed and double primed copies of a state $q$ are controlled if, and only if, $q$ was controlled in $\sys{M}$.
    The double primed states have colour $0$, while all original and primed states have colour $1$.
    See \cref{fig:storage-PosMP} (on the left) for an illustration.

To give the idea of this construction in a nutshell, the B\"uchi condition in $\sys{M'}$ intuitively sells one energy unit for visiting an accepting state (or: for visiting a state with colour $0$, the double primed copy).
$\ES{k}$ implies that, as soon as $s+1$ energy is available, one can sell off one energy unit for a visit of an intermediate accepting state. $\PosMP$ implies that this can almost-surely be done infinitely often.
Vice-versa, $\ES{k}$ implies non-negative mean payoff.
$\ES{k}$ plus B\"uchi can always be realised with finite memory by \cref{thm:as-storage-parity} (2),
and such a strategy then implies that $\PosMP\cap\ES{k}$ holds almost-surely in $\sys{M}$.
Now the claim holds by \cref{thm:as-storage-parity}.
\end{proof}

\begin{remark}
Note that the order of quantification in the limit-sure problems for storage objectives
($\exists s. \forall \eps \ldots$)
means that limit-sure and almost-sure winning coincides for storage-parity objectives:
if there is an $s$ such that $\ES{s,s} \cap \Parity{}$ holds limit-surely
then one can get rid of the storage condition by hardcoding energy-values up to $s$ into the states.
The same is true for mean-payoff-storage objectives.
The claims in \cref{thm:as-storage-parity,thm:as-ES-PosMP} thus also hold for the limit-sure problems.
\end{remark}

Finally, we restate the result from \cite{CD2011}, Theorem~2 (1) on positive mean-payoff-parity objectives 
and add to it an explicit computational complexity bound that we will need later.

\begin{theorem}[Mean-payoff-Parity]
    \label{thm:as-Par-PosMP}
        For finite MDPs with combined positive mean-payoff and parity objectives,
    \begin{itemize}
        \item[1)] The almost-sure problem is in {\sc P} and can be solved in 
        $\widetilde {\mathcal O}(d\cdot|V|^{3.5} \cdot (|\lambda| + |\cost{}|)^2)$ time.
        \item[2)] Finite memory suffices.
    \end{itemize}
\end{theorem}
\begin{proof}
The computation complexity bound follows from the analysis of Algorithm 1 in \cite{CD2011}. 
It executes $d/2$ iterations of a loop, in which Step 3.3 of computing the mean-payoff of maximal end components dominates the cost. This can be formulated as a linear program (LP) that uses two variables, called {\em gain} and {\em bias}, for each state\cite{Puterman:book}. This LP can be solved using Karmarkar's algorithm \cite{Karmarkar/84/Karmarkar} in time
      $\widetilde {\mathcal O}(|V|^{3.5} \cdot (|\lambda| + |\cost{}|)^2)$. Note that the complexity refers to \emph{all} (not each) maximal end-components.

As we do not need to obtain a maximal payoff $\pi>0$ but can use any smaller value, like $\pi/2$, finite memory suffices.
\end{proof}

\section{Almost-Sure Energy-Parity}
\label{sec:as-EP}
\label{sec:correction}
\newcommand{\annoy}{\\ }
In this section we prove \cref{thm:correction}.
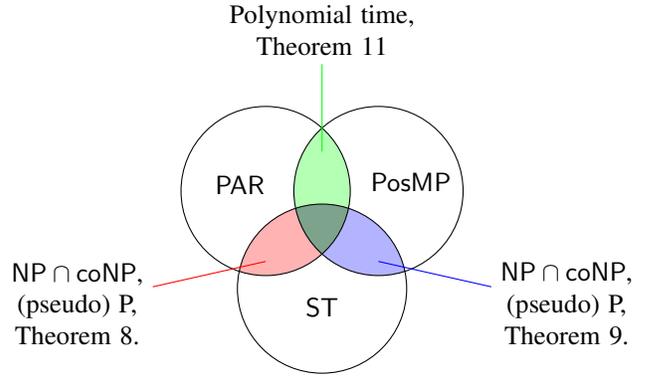
\begin{figure}[t]
  \centering
\def\firstcircle{(0,0) circle (1.5cm)}
\def\secondcircle{(-60:2cm) circle (1.5cm)}
\def\thirdcircle{(0:2cm) circle (1.5cm)}
\begin{tikzpicture}[scale=0.75]
    \begin{scope}[fill opacity=0.3]
    \begin{scope}
       \clip \firstcircle;
       \fill[red] \secondcircle;
    \end{scope}
    \begin{scope}
       \clip \secondcircle;
       \fill[blue] \thirdcircle;
    \end{scope}
    \begin{scope}
       \clip \thirdcircle;
       \fill[green] \firstcircle;
    \end{scope}
  \end{scope}
    \draw \firstcircle node [label={[label distance=-0.4cm]120:$\Parity{}$}] {};
    \draw \secondcircle node [anchor=north] {$\mathsf{ST}$};
    \draw \thirdcircle node [label={[label distance=-0.4cm]30:$\PosMP$}] {};
  \draw (4,-2) node[right,align=center]
      {
          $\NP \cap\coNP$,\annoy
          (pseudo) P,\annoy
          \cref{thm:as-ES-PosMP}.
      } 
      edge[blue] (2,-1.25);
  \draw (-2,-2) node[left,align=center] {
          $\NP \cap\coNP$,\annoy
          (pseudo) P,\annoy
          \cref{thm:storage}.
  } edge[red] (0,-1.25);
  \draw (1,2.25) node[above,align=center] {
Polynomial time,\annoy
\Cref{thm:as-Par-PosMP}
} edge[green] (1,0.7);
\end{tikzpicture}
\caption{
\label{fig:combinations}
Results on the almost-sure problems for
combined objectives.
Positive mean-payoff-parity 
${\PosMP\cap\Parity{}}$ (in green)
requires infinite memory
while
storage-parity (in red)
and 
storage-mean-payoff (in blue)
have finite-memory optimal strategies.
} 
\end{figure}
Our proof can be explained in terms of the three basic objectives:
storage ($\mathsf{ST}$), positive mean-payoff ($\PosMP$), and parity ($\Parity$).
It is based on the intuition provided by the counterexample in the previous section.
Namely, in order to almost-surely satisfy the energy-parity condition one needs to combine two strategies:
\begin{enumerate}
    \item One that guarantees the parity condition and, at the same time, a positive expected mean-payoff.
        Using this strategy one can achieve the energy-parity objective with some non-zero chance.
    \item A \emph{bailout} strategy that guarantees positive expected mean-payoff together with a storage condition.
        This allows to (almost-surely) set the accumulated energy level to some arbitrarily high value.
\end{enumerate}
We show that, unless there exist some safe strategy that satisfies storage-parity,
it is sufficient (and necessary) that such two strategies exist
and that the controller can freely switch between them.
I.e.\ they do not leave the \emph{combined} almost-sure winning set unless a state that satisfies storage-parity is reached.

\newcommand{\Bail}{\PosMP\cap\ES{k}}

\medskip
Recall that the combined positive mean-payoff-parity objective (for case 1 above) 
is independent of an initial energy level
and its almost-sure problem is decidable in polynomial time due to \Cref{thm:as-Par-PosMP}.
The mean-payoff-storage objective $\ES{k}\cap\PosMP$ (for case 2 above),
as well as the storage-parity objective are computable by \cref{thm:as-ES-PosMP,thm:storage}, respectively.
See \cref{fig:combinations}.

\medskip
To establish \cref{thm:correction},
we encode the almost-sure winning sets of the storage-parity objective directly into the system
(\cref{def:encode-sc,lem:ff-ext-claims}),
in order to focus on the two interesting conditions from above.
We then show
(\cref{def:fix:refinement,lem:safe-computable})
that the existence of the two strategies for bailout and $\mathsf{ST} \cap \PosMP$, and the minimal safe energy levels
can be computed in the claimed bounds.
In \cref{lem:mainfix} we show that these values coincide with the minimal energy levels of the energy-parity objective for the original system, which concludes the proof.

\begin{definition}
    \label{def:encode-sc}
    For a given MDP $\sys{M}$ and associated $\cost$ and $\parity$ functions, we define
    an MDP $\sys{M}'\eqdef(\VC',\VP', E', \prob')$ with states $V'\eqdef \VC'\uplus\VP'$
    as follows.
    For every state $q$ of $\sys{M}$
    there are two states, $q$ and $q'$ in $V'$ such that both have the same colour
    as $q$ in $\sys{M}$, every original incoming edge now only goes to $q'$,
    and
    every original outgoing edge now only goes from $q$.
    Moreover, $q'$ is controlled and has an edge to $q$ with $\cost(q',q)=0$.

    Finally, $\sys{M}'$ contains a single winning sink state $w$ with colour
    $0$ and a positive self-loop,
    and every state $q'$ gets an edge to $w$ where the cost 
    of $-k_q$, where $k_q\in\N$ is the minimal value such that, for some $s\in\N$,
    the storage-parity objective $\ES{k_q,s}\cap\Parity$ holds almost-surely
    See \cref{fig:encode-sc} (on the right) for an illustration.
\end{definition}

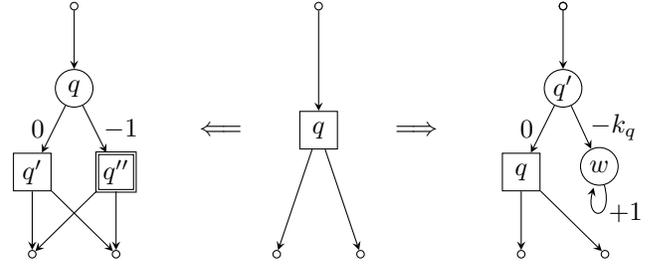
\begin{figure}[t]
  \centering
  \begin{tikzpicture}[node distance=1.65cm and 3.25cm, on grid]
    \tikzstyle{acstate}=[cstate,minimum size=0.1cm]
    \node[pstate] (q) {$q$};
    \node[acstate, above=of q] (top) {};
    \coordinate [below=of q] (midbelow) ;
    \node[acstate, left=0.5cm of midbelow] (belowl) {};
    \node[acstate, right=0.5cm of midbelow] (belowr) {};
    \draw[->] (top) to (q);
    \draw[->] (q) to (belowr);
    \draw[->] (q) to (belowl);

    \begin{scope}
    \node[acstate, right=of top] (top') {};
    \coordinate [right= of q] (mid') ;
    \node[cstate, above=0.3cm of mid'] (q') {$q'$};
    \node[acstate, right=of top] (top') {};
    \node[acstate, right=of belowl] (belowl') {};
    \node[acstate, right=of belowr] (belowr') {};
    \node[pstate, below left=0.3cm and 0.3cm of mid'] (q'') {$q$};
    \node[cstate, below right=0.3cm and 0.3cm of mid'] (w) {$w$};
    \draw[->] (top')to (q');
    \draw[->] (q') to node[left] {$0$} (q'');
    \draw[->] (q'') to (belowr');
    \draw[->] (q'') to (belowl');
    \draw[->] (q') to node[right] {$-k_q$} (w);
    \draw[->,loop below] (w) edge node[right] {$+1$} (w);
 
    \node (b) at ($(q)!0.4!(mid')$) {$\implies$};
    \end{scope}

    \begin{scope}
    \node[acstate, left=of top] (top') {};
    \coordinate [left= of q] (mid') ;
    \node[acstate, left=of belowl] (belowl') {};
    \node[acstate, left=of belowr] (belowr') {};
    \node[cstate, above=0.3cm of mid'] (q') {$q$};
    \node[pstate, below left=0.3cm and 0.3cm of mid'] (q'') {$q'$};
    \node[pstate, double,below right=0.3cm and 0.3cm of mid'] (q''') {$q''$};
    \draw[->] (top')to (q');
    \draw[->] (q') to node[left] {$0$} (q'');
    \draw[->] (q') to node[right] {$-1$} (q''');
    \draw[->] (q'') to (belowr');
    \draw[->] (q'') to (belowl');
    \draw[->] (q''') to (belowr');
    \draw[->] (q''') to (belowl');
    
    \node (b) at ($(q)!0.4!(mid')$) {$\impliedby$};
    \end{scope}
\end{tikzpicture}
\caption[ The reduction from \cref{thm:as-storage-PosMP,def:encode-sc}.]{
\label{fig:storage-PosMP}
\label{fig:encode-sc}
The reduction from \cref{thm:as-storage-PosMP} (middle to left)
makes sure that controller can always trade
energy values for visiting an accepting state.\\
The reduction from \cref{def:encode-sc} (middle to right)
makes sure that in each state $q$ the controller
can interrupt and go to a winning sink
if the accumulated energy level exceeds
the minimal value $k_q$ necessary to almost-surely win the storage-parity objective from $q$.
} 
\end{figure}
The next lemma summarises the relevant properties of $\sys{M}'$.
It follows directly from \cref{def:encode-sc} and the observation that $\ES{k,s}\subseteq \EN{k}$ holds for all $k,s\in\N$.

\begin{lemma}
    \label{lem:ff-ext-claims}
    For every MDP $\sys{M}$, state $q$ and $k \leq s\in\N$,
    \begin{enumerate}
        \item $\EN{k}\cap\Parity$ holds almost-surely in $\sys{M}$ if, and only if,
        it holds almost-surely in $\sys{M'}$.
        \item If $\ES{k,s}\cap\Parity{}$ holds almost-surely in $\sys{M}$
        then $\ES{k,s}\cap\Parity{}\cap\PosMP$ holds almost-surely in $\sys{M}'$.
    \end{enumerate}
\end{lemma}

For a set $S\subseteq V'$ of states, we write $\sys{M}'|S$ for the restriction of $\sys{M}'$ to states in $S$,
i.e.\ the result of removing all states not in $S$ and their respective edges.
\newcommand{\safe}[1]{\text{safe}(#1)}
\begin{definition}
    \label{def:fix:refinement}
We define $R\subseteq V'$ as the largest set of states such that, in $\sys{M}'|R$, every state
\begin{enumerate}
    \item almost-surely satisfies the $\PosMP\cap\Parity$ objective, and
    \item almost-surely satisfies the $\PosMP\cap\mathsf{ST}$ objective.
\end{enumerate}
For every state $q\in V'$ let $\safe{q}\in\N$ be the minimal number $k$ such that
$\PosMP\cap\ES{k}$ holds almost-surely in $\sys{M}'|R$
and $\safe{q}\eqdef\infty$ if no such number exists (in case $q\notin R$).
\end{definition}

The relevance of these numbers for us is, intuitively, that if $\safe{q}$ is finite,
then there exists a pair of strategies, one for the $\PosMP\cap\Parity$ and one for the $\Bail$ objective,
between which the controller can switch as often as she wants.

\begin{lemma}
    \label{lem:safe-computable}
    For a given $\sys{M}'$,
    the values $\safe{q}$ are either $\infty$ or
    bounded by ${\mathcal O}(\card{V}\cdot W)$, computable in pseudo-polynomial time
    $\widetilde {\mathcal O}(d\cdot|V|^{4.5} \cdot (|\lambda| + |\cost{}|)^2 +
    \card{E}\cdot d\cdot \card{V}^5\cdot W)$,
    and verifiable in $\NP\cap\coNP$.
\end{lemma}
\begin{proof}
    Finite values $\safe{q}\in \N$ are clearly bounded
    by minimal sufficient values for almost-sure satisfiability of $\ES{k}\cap\PosMP$
    in $\sys{M'}$. Therefore, the claimed bound holds by definition of $\sys{M}'$
    and \cref{thm:as-storage-parity,thm:as-storage-PosMP}.

    The set $R$ is in fact the result of a refinement procedure that starts with all states
    of $\sys{M}'$.
    In each round, it removes states that fail either of the two conditions.
    For every projection $\sys{M}'|S$, checking Condition 1 takes
    $\widetilde {\mathcal O}(d\cdot|V|^{3.5} \cdot (|\lambda| + |\cost{}|)^2)$
    time by \cref{thm:as-Par-PosMP}
    and Condition 2 can be checked in
    $\mathcal{O}(\card{E}\cdot d\cdot \card{V}^4\cdot W)$
    time
    by \cref{thm:as-ES-PosMP}.
    All in all, this provides a pseudo-polynomial time algorithm to compute $R$.
    By another application of \cref{lem:bailout-to-parity}, we can compute the (pseudo-polynomially bounded)
    values $\safe{q}$.
    In order to verify candidates for values $\safe{q}$ in $\NP$, and also $\coNP$,
    one can guess a witness, the sequence of sets $R_0\supset R_1 \supset \ldots \supset R_j=R$,
    together with certificates for all $i\le j$ that $R_{i+1}$ is the correct set following $R_i$
    in the refinement procedure.
    This can be checked all at once by considering the disjoint union of all $\sys{M}'|R_i$.
\end{proof}

\begin{lemma}
    \label{lem:mainfix}
    For every $k\in\N$ and state $q$ in $\sys{M}'$,
    the energy-parity objective
    $\EN{k}\cap\Parity$ holds almost-surely from $q$ if, and only if, $\safe{q}\le k$.
\end{lemma}
\begin{proof}
    {\bf ($\implies$)}. 
    First observe that the winning sink $w$ in $\sys{M}'$ is contained in $R$,
    and has $\safe{w}=0$ since the
    only strategy from that state satisfies
    $\ES{0,0}\cap\Parity{}\cap\PosMP$.

    For all other states there are two cases: either there is an $s\in\N$ such that
    $\ES{k,s}\cap\Parity{}$ holds almost-surely, or there is no such $s$.
    If there is, then the strategy that goes to the sink guarantees the objective
    $\ES{k,s}\cap\Parity{}\cap\PosMP$, which implies the claim.

    For the second case (there is no $s$ such that $\ES{k,s}\cap\Parity{}$ holds almost-surely)
    we see that every almost-surely winning strategy
    for $\EN{k}\cap\Parity{}$
    must also almost-surely satisfy $\PosMP$.
    To see this, note that the energy condition implies a non-negative expected mean-payoff,
    and that an expected mean-payoff of $0$
    would imply that the storage condition
    $\ES{k,s}$ is satisfied for some $s$,
    which contradicts our assumption.
    Consequently the $\PosMP\cap\Parity{}$ objective holds almost-surely.

    We now show that the $\ES{k,s}\cap\PosMP$ objective holds almost-surely in state $q$, where $s>\safe{r}$ for all states $r$ with $\safe{r} < \infty$. 
    We now define a strategy that achieves $\ES{k,s}\cap\PosMP$.
    For this, we first fix a strategy $\sigma_q$ that achieves $\EN{h_q} \cap \Parity{}$ with $h_q=\safe{q}$ for each state $q$ with $\safe{q} < \infty$.
    
    When starting in $q$, we follow $\sigma_q$ until one of the following three events happen.
    We have (1) sufficient energy to move to the winning sink $w$. In this case we do so.
    Otherwise, if we (2) have reached a state $r$ and since starting to follow $\sigma_q$, the energy balance is strictly greater than
    \footnote{Note that the energy balance can never be strictly smaller than $h_r - h_q$ in such a case, as there would not be a safe continuation from $r$ otherwise.} $h_r - h_q$.
    Then we abandon $\sigma_q$ and follow $\sigma_r$ as if we were starting the game.
    
    Before we turn to the third event, we observe that, for each strategy $\sigma_q$, there is a minimal distance
    \footnote{If there is no point where (2) is met, the energy balance on state $r$ is always exactly $h_r - h_q$, such that $\sigma_q$ satisfies $\ES{h_q,s}$, and (1) is satisfied immediately.}
    $d_q \in \mathbb N$ to (1) or (2) and a positive probability $p_q > 0$ that either event is reached in $d_q$ steps.
    The third event is now simply that (3) $d_q$ steps have lapsed.
    When in state $r$ we then also continue with $\sigma_r$ as if we were starting the game.
    
    It is obvious that no path has negative mean payoff.
    Moreover, as long as the game does not proceed to the winning sink, a partial run starting at a state $q$ and ending at a state $r$ has energy balance $\geq h_r - h_q$, such that the resulting strategy surely satisfies $\mathsf{ST}$.
    The expected mean payoff is $\ge p_q/d_q$, and $\PosMP$ is obviously satisfied almost-surely.
    Consequently, $\ES{h_q,s}\cap\PosMP$ holds almost-surely from $q$.

    We conclude that every state for which the $\EN{k}\cap\Parity{}$ objective holds almost-surely
    must satisfy both criteria of \cref{def:fix:refinement} and thus be a member of $R$.
    Since almost-sure winning strategies cannot leave the respective winning sets, this means that
    every winning strategy for the above objective also applies in $\sys{M'}|R$ and thus justifies
    that $\safe{q}\le k$.

    {\bf ($\impliedby$)}.
    By definition of $R$, there are two finite memory strategies $\sigma$ and $\beta$ which
    almost-surely satisfy the $\PosMP\cap\Parity$, and the bailout objective $\PosMP\cap\ES{k}$, respectively,
    from every state $q$ with $\safe{q}\le k$. Moreover, those strategies will never visit any state
    outside of $R$. 
    
    We start with the bailout strategy $\beta$ and run it until the energy level is high enough (see below).
    We then turn to $\sigma$ and follow it until (if ever) it \emph{could} happen in the next step that a state $q$ is reached while the energy level falls below $\safe{q}$. We then switch back to $\beta$.
    
    The ``high enough'' can be achieved by collecting enough energy that there is a positive probability that one does not change back from $\sigma$ to $\beta$. For this, we can start with a sufficient energy level $e$ such that $\sigma$ never hits an energy $\leq 0$ with a positive probability
\footnote{We argue in \Cref{lem:MPge0} that, with sufficient energy, this probability can be moved arbitrarily close to $1$.}. The sum $e+s+W$ consisting of this energy, the sufficient storage level for $\PosMP \cap \ES{k}$, and the maximal change $W$ of the energy level obtained in a single step suffices.
    
    The constructed strategy then almost-surely satisfies the $\EN{k_q}\cap\PosMP\cap\Parity$ objective
    from every state $q$ and $k_q\eqdef\safe{q}$.
    In particular, this ensures that the $k$-energy-parity objective holds almost-surely from $q$
    in $\sys{M}'|R$ and therefore also in $\sys{M}'$.
\end{proof}

\begin{proof}[Proof of \cref{thm:as-energy-parity}]
(1) The fact that infinite memory is necessary follows
from our counterexample to \cref{claim}, and the observation that
every finite memory winning strategy for the $\Parity$ objective must also be colour-committing.

For parts (2) and (3), it suffices, by \cref{lem:ff-ext-claims}(1) and \cref{lem:mainfix},
to construct $\sys{M}'$ and compute the values $\safe{q}$ for every state $q$ of $\sys{M}'$.
The claims then follow from \cref{lem:safe-computable}.
\qedhere
\end{proof}

\section{Limit Values}
\label{sec:limit}
Since $\EN{k}\subseteq\EN{k+1}$ holds for all $k\in\N$, the chance of satisfying the $k$-energy-parity objective
depends (in a monotone fashion) on the initial energy level: for every state $p$ we have that
$
\Val[\sys{M}]{p}{\EN{k}\cap\Parity{}} \le \Val[\sys{M}]{p}{\EN{k+1}\cap\Parity{}}.
$
We can therefore consider
the respective \emph{limit values} as the limits of these values for growing $k$:
\begin{equation}
    \LValP[\sys{M}]{p} \quad \eqdef \quad \sup_k \Val[\sys{M}]{p}{\EN{k}\cap\Parity}.
\end{equation}
Note that this is \emph{not} the same as the value of $\Parity$ alone.
For instance, the state $l$ from \cref{ex:lval} has limit value $\LValP{l}=0
\neq \Val{l}{\Parity{}}=1$.

The states $r$ and $s$ from \cref{ex:lval} have $\LValP{r}{=}1$ and $\LValP{s}{=}1$.
In fact, for any $\sys{M},w,k$ and parity objective $\Parity$ it holds that
$\Val[\sys{M}]{p}{\EN{k}\cap\Parity} ~\le~ \LValP[\sys{M}]{p} ~\le~ \Val[\sys{M}]{p}{\Parity}$.

\medskip
Limit values are an important ingredient in our proof of \cref{thm:main}.
This is due to the following property, which directly follows from the definition.
\begin{lemma}
    \label{lem:value-one-states}
    Let $\sys{M}$ be an MDP and $p$ be a state with $\LValP[\sys{M}]{p} = 1$.
    Then, for all $\eps>0$, there exist a $k\in\N$ and a strategy $\sigma$ such that
    $\Prob[\sys{M}(\sigma)]{p}(\EN{k}\cap\Parity{}) \ge 1- \eps$.
\end{lemma}

We now show how to compute limit values,
based on the following two sets.
\begin{align*}
    A &\quad\eqdef\quad \{p\in Q\mid \exists k \exists \sigma~\Prob[\sys{M}(\sigma)]{p}(\EN{k}\cap\Parity{})=1\}\\
    B &\quad\eqdef\quad \{p\in Q\mid \exists \sigma~\Prob[\sys{M}(\sigma)]{p}(\PosMP\cap\Parity{})=1\}
\end{align*}

The first set, $A$, contains those states that satisfy the $k$-energy-parity condition
almost-surely for some energy level $k\in\N$.
The second set, $B$, contains those states that almost-surely satisfy the combined
positive mean-payoff-parity objective.
Our argument for computability of the limit values is based on the following theorem,
which claims that limit values correspond to the values of a reachability objective
with target $A\cup B$. Formally,

\begin{theorem}
    \label{thm:LVAL-reach}
    For every MDP $\sys{M}$ and state $p$,
    $\LValP[\sys{M}]{p} = \sup_{\sigma} \Prob[\sys{M}(\sigma)]{p}(\eventually (A\cup B))$.
\end{theorem}

Before we prove this claim by \cref{lem:LVAL-reach-up,lem:LVAL-reach-low} in the remainder of this section, we remark that we can compute $A \cup B$ without constructing $A$.
Let us consider the set
\[
    A' \quad\eqdef\quad \{p\in Q\mid \exists k \exists \sigma~\Prob[\sys{M}(\sigma)]{p}(\ES{k}\cap\Parity{})=1\}\ ,
\]
and observe that $A' \subseteq A$ holds by definition and that the construction of $A$ from \cref{thm:as-energy-parity} establishes $A \subseteq A' \cup B$.
Thus, $A \cup B = A' \cup B$ holds, and it suffices to construct $A'$ and $B$, which is cheaper than constructing $A$ and $B$.

We now start with some notation.

\newcommand{\Att}[1]{\mathit{Att}(#1)}
\begin{definition}
For an MDP $\sys{M}\eqdef(\VC,\VP, E, \prob)$,
the \emph{attractor} of a set $X\subseteq V$ of states
is the set
$\Att{X} \eqdef \{q \mid \exists \sigma~\Prob[\sys{M}(\sigma)]{q}(\eventually X)=1\}$
of states that almost-surely satisfy the reachability objective with target $X$.
\end{definition}

\begin{definition}
For an MDP $\sys{M}\eqdef(\VC,\VP, E, \prob)$
an \emph{end-component} is a strongly connected set of states $C \subseteq V$
with the following closure properties:
\begin{itemize}
 \item for all controlled states $v \in C \cap V_C$, some successor $v'$ of $v$ is in $C$, and
 \item for all probabilistic states $v \in C \cap V_P$, all successors $v'$ of $v$ are in $C$.
\end{itemize}
Given $\cost{}$ and $\parity{}$ functions and $i\in\N$, we call an end-component
\begin{itemize}
    \item \emph{$i$ dominated}, if they contain a state $p$ with $\parity(p)=i$, but no state $q$ with $\parity(q)<i$,
    \item \emph{$i$ maximal}, if it is a maximal (w.r.t.~set inclusion) $i$ dominated end-component, and
    \item \emph{positive}, if its expected mean-payoff is strictly greater than $0$ (recall that the mean-payoff of all states in a strongly connected set of states of an MDP is equal).
\end{itemize}
\end{definition}
\begin{restatable}{lemma}{LemMPge}
\label{lem:MPge0}
The states of each positive $2i$-maximal end-component $C$ are contained in $B$.
\end{restatable}
\begin{proof}(sketch)
We consider a strategy $\sigma$ that follows the optimal (w.r.t. the mean-payoff) strategy most of the time and ``moves to'' a fixed state p with the minimal even parity 2i only sparsely. Such a strategy keeps the mean-payoff value positive while satisfying the parity condition. We show that $\sigma$ can be defined to use finite memory or no memory, but randomisation. Either way, $\sigma$ induces a probabilistic one-counter automata \cite{etessami2010quasi}, whose probability of ever decreasing the counter by some finite $k$ can be analysed,
based on the mean-payoff value, using the results in \cite{BKK/14}. 
\qedhere
\end{proof}

\begin{lemma}
    \label{lem:LVAL-reach-up}
    For every MDP $\sys{M}$ and state $p$,
    $\LValP[\sys{M}]{p} \ge \sup_{\sigma} \Prob[\sys{M}(\sigma)]{p}(\eventually (A\cup B))$.
\end{lemma}
\begin{proof}
Assume w.l.o.g.\ that
$\pi \eqdef \sup_{\sigma} \Prob[\sys{M}(\sigma)]{p}(\eventually (A\cup B)) > 0$.
We show that $\LValP[\sys{M}]{p}$ is at least $\pi - 2\varepsilon$ for all $\varepsilon>0$ as follows.

We start by choosing $k\in\N$ big enough so that 
for every state $q\in A\cup B$, some strategy satisfies the $k$-energy-parity objective with probability $>1-\eps$.
We then consider a memoryless strategy (e.g.\ from solving the associated linear program), which guarantees that the
set $A\cup B$ is reached with likelihood $\pi$, and then determine a natural number $l$ such that it is
reached within $l$ steps with probability $>\pi - \eps$.
This reachability strategy $\sigma$ can now be combined with an $\eps$-optimal strategy for states in $A\cup B$:
until a state in $A\cup B$ is reached, the controller plays according to $\sigma$
and then swaps to a strategy that guarantees the $k$-energy-parity objective with likelihood $>(1-\eps)$.
Such a strategy exists by our assumption on $k$.
This combined strategy will satisfy the $\EN{k+l}$-energy-parity objective
with probability $> (\pi-\eps)(1-\eps) \ge \pi - 2\eps$.
\qedhere
\end{proof}

\begin{definition}\label{def:non-losing}
[Non-losing end-component]
We call an end-component non-losing, iff the smallest priority of states in the end-component is even and there is a strategy that allows to
\begin{enumerate}
 \item almost-surely stay within this end-component,
 \item almost-surely visit all states in the end-component, and
 \item satisfy the energy condition from some energy level with non-zero probability.
\end{enumerate}
\end{definition}
\begin{lemma}
\label{lem:non-losing}
Every non-losing end-component $I$ is contained in $\Att{A\cup B}$.
\end{lemma}
\begin{proof}
We start with a case distinction of the mean-payoff value of $I$. (Recall that, as an end-component in an MDP, all states of $I$ have the same mean-payoff values.)
If this value is positive and $2i$ is the lowest priority in $I$,
then $I$ is contained in some $2i$ maximal end-component and by \cref{lem:MPge0}, also in $B\subseteq\Att{A\cup B}$.
If this value is negative, then the third condition of \cref{def:non-losing} cannot be satisfied together with the first two. 
This leaves the case where the mean-payoff value is~$0$.

If the mean-payoff value is $0$, then there exists a bias function $b:I\to \Z$ that satisfies the following constraints:
\begin{itemize}
 \item $b(v) = \min\Big\{\mathsf{cost}\big((v,v')\big) + b(v') \mid v' \in I \wedge (v,v') \in E\Big\}$ holds for all controlled states $v \in I \cap V_C$,
\item
$b(v) = \sum\limits_{v' \in \{w \in I \mid (v,w) \in E\}}   
    \lambda(v)\big((v,v')\big) \cdot \Big(\mathsf{cost}\big((v,v')\big)  + b(v') \Big)$
for all probabilistic states $v \in I \cap V_P$.
\end{itemize}

When adjusting $b$ to $b'$ by adding the same constant to all valuations, $b'$ is obviously a bias function, too. 

We call a transition $(v,v')$ \emph{invariant} iff $b(v) = \mathsf{cost}\big((v,v')\big) + b(v')$ holds.
A set $G\subseteq V$ of states invariant if it is strongly connected and contains only controlled states with an invariant transition into $G$ and only probabilistic states with only invariant outgoing transitions, which all go to $G$.
We now make the following case distinction.

Case 1:
there is a nonempty, invariant set $G \subseteq I$, such that the state $p$ of $G$ with minimal priority has even priority.
First notice that $G\subseteq A$:
if the minimal value of the bias function is $b_{\min}$, then the bias of a state in $p$ minus $b_{\min}$ serves as sufficient energy when starting in $p$:
it then holds that $\Prob[\sys{M}(\sigma)]{p}(\EN{k}\cap\Parity{})=1$, where $k\eqdef b(p) - b_{\min}$,
and $\sigma$ is a memoryless randomised strategy that assigns a positive probability to all transitions into $G$.
Since $I$ is an end-component, it is contained in the attractor of $G$, which implies the claim, as
$\Att{G} \subseteq \Att{A} \subseteq \Att{A\cup B}$.

Case 2: there is no non-empty invariant set $G \subseteq I$ with even minimal priority.
We show that this is a contradiction with the assumption that $I$ is a non-losing set,
in particular with condition 3 of \cref{def:non-losing}.
We assume for contradiction that there is a strategy $\sigma$ and an energy level $k$ such that we can satisfy the energy parity condition with a positive probability 
while staying in $I$ and starting at some state $p \in I$.
We also assume w.l.o.g.\ that all bias values are non-negative, and $m$ is the maximal value among them.
We set $k' = k + m$.

The `interesting' events that can happen during a run are selecting a non-invariant transition from a controlled state or reaching a probabilistic state (and making a random decision from this state), where at least one outgoing transition is non-invariant.

We capture both by random variables, where
random variables that refer to taking non-invariant transition from controlled states (are deterministic and) have a negative expected value, while
random variables that refer to taking a transition from a probabilistic state where at least one outgoing transition is non-invariant refers to a random variable drawn from a finite weight function with expected value $0$ and positive variation. 
Note that random variables corresponding to probabilistic non-invariant transitions are independent and
drawn from a \emph{finite} set of distributions. 

Let $\alpha$ be any infinite sequence of such random variables.
From the results on finitely inhomogeneous controlled random walks \cite{durrett1991making}, we can show that
almost-surely the sum of some prefix of $\alpha$ will be lower than $-k'$ (and in fact lower than any finite number). 
The proof follows the same reasoning as in Proposition 4.1 of \cite{BBEKW10}, where a sufficient and necessary condition was given for not going bankrupt with a positive probability in \emph{solvency games} \cite{BergerKSV08}.

We now consider the set of runs induced by $\sigma$.
As we just showed, almost all runs that have infinitely many interesting events (as described above) will not satisfy the $k'$-energy condition. 
Almost all runs that have finitely many interesting events will 
have an odd dominating priority, and therefore will not satisfy the parity condition. 
Thus, the probability that the energy parity condition is satisfied by $\sigma$ is $0$.
\qedhere
\end{proof}
\begin{lemma}
    \label{lem:LVAL-reach-low}
    For every MDP $\sys{M}$ and state $p$,
    $\LValP[\sys{M}]{p} \le \sup_{\sigma} \Prob[\sys{M}(\sigma)]{p}(\eventually (A\cup B))$.
\end{lemma}
\begin{proof}
    Fix $p$ and $\sigma$.
    Every run from $p$ will, with probability $1$, eventually reach an end-component and 
    visit all states of the end-component infinitely often \cite{de1997formal}.

    Let $C$ be an end-component such that $C$ forms the infinity set of the runs from $p$ under $\sigma$ with a positive probability $\pi>0$. If $C$ does not satisfy the conditions of non-losing end-components, then
    the probability $\Prob[\sys{M}(\sigma)]{q}(\EN{k}\cap\Parity{})$
    that the $k$-energy-parity objective is satisfied from some state $q\in C$ is $0$, independent of the value $k$.
    Thus, the probability of satisfying the $k$-energy-parity objective from an initial state $p$
    is bounded by the chance of reaching a state in some non-losing end-component.
    These observations hold for every strategy $\sigma$
    and therefore we can bound
    \begin{align*}
        \LValP[\sys{M}]{p} &=
        \sup_{k}\sup_{\sigma}\Prob[\sys{M}(\sigma)]{p}(\EN{k}\cap\Parity)\\
                           &\le\sup_{\sigma} \Prob[\sys{M}(\sigma)]{p}(\eventually (\mathit{NLE})),
    \end{align*}
    where $\mathit{NLE}\subseteq V$ denotes the union of all non-losing end-components.
    Now \cref{lem:non-losing} implies that
    $\sup_{\sigma} \Prob[\sys{M}(\sigma)]{p}(\eventually (\mathit{NLE}))
    \leq \sup_{\sigma} \Prob[\sys{M}(\sigma)]{p}(\eventually (A\cup B))$,
    which completes the proof.
\qedhere
\end{proof}
\begin{lemma}
    \label{lem:checking-LVAL-1}
    Determining the limit value of a state $p$
    can be done in $\widetilde {\mathcal O}(|E|\cdot d\cdot \card{V}^4\cdot W + d \cdot |V|^{3.5} \cdot (|\lambda| + |\cost{}|)^2)$ deterministic time.
    They can also be determined in \NP\ and \coNP\ in the input size when $W$ is given in binary.
\end{lemma}
\begin{proof}
    Recall that $\LValP[\sys{M}]{p} = \sup_{\sigma} \Prob[\sys{M}(\sigma)]{p}(\eventually (A\cup B))$
    by Theorem~\ref{thm:LVAL-reach}, that $A \cup B = A' \cup B$, and that
    $A'$ and $B$ are the sets of control states
    that almost-surely satisfy the storage-parity and mean-payoff-parity
    objective, respectively.
    Using the results of Section~\ref{sec:limit}, the algorithm proceeds as
    follows.
    \begin{enumerate}
     \item Compute $A'$, which can be done in time $\mathcal O(|E|\cdot d \cdot
    |V|^4 \cdot W)$ by \cref{thm:as-storage-parity}.
     \item Compute, for each occurring even priority $2i$, the following:
     \begin{enumerate}
      \item the set of $2i$ maximal end-components, which can be computed in $\mathcal O(|E|)$; and
      \item the mean payoff value for the $2i$ maximal end-components
      can be computed using Karmarkar's algorithm \cite{Karmarkar/84/Karmarkar} for linear programming in time     $\widetilde {\mathcal O}(|V|^{3.5} \cdot (|\lambda| + |\cost{}|)^2)$
       ---note that the complexity refers to \emph{all} (not each) $2i$ maximal end-components.
    \end{enumerate}
    \item Consider the union of $A$ with all the $2i$ maximal end-components with positive mean payoff
    computed in Step 2, and compute the maximal achievable probability of
    reaching this set. (By the results of Section~\ref{sec:limit}, this yields
    the probability $\sup_{\sigma} \Prob[\sys{M}(\sigma)]{p}(\eventually (A\cup B))$.)
    \end{enumerate}
The last step costs $\widetilde {\mathcal O}(|V|^{3.5} \cdot |\lambda|^2)$ \cite{Karmarkar/84/Karmarkar} for solving the respective linear program \cite{Puterman:book}, which is dominated by the estimation of the cost of solving the linear programs from (2b).
Likewise, the cost of Step (2a) is dominated by the cost of Step (2b).

This leaves us with once the complexity of (1) and $d$ times the complexity of (2b), resulting in the claimed complexity.
Note that it depends on the size of representation (in binary) $\lambda$ and $W$ (in unary), and the bigger of these values dominates the complexities.

Finally, all steps are in \NP\ and in \coNP.
    \qedhere
\end{proof}

\section{Limit-Sure Energy-Parity}
\label{sec:ls-EP}
\label{sec:extensions}
In this section we
provide the reduction from checking if an energy-parity objective
holds limit-surely, to checking if such an objective holds almost-surely.
The reduction basically extends the MDP so that the controller may
``buy'' a visit to a good priority (at the expense of energy)
if currently in a state $p$ with limit value $\LValP{p}=1$.

\begin{definition}
    The \emph{extension} of a finite MDP $\sys{M}$ for given cost and parity functions
    is the MDP $\sys{M'}\supseteq\sys{M}$ where, additionally, for every controlled state
    $s\in\VC$ with $\LValP[\sys{M}]{s}=1$,
    there is a new state $s'$ with parity $0$
    and edges $(s,s'), (s',s)$ with 
    $\cost(s,s') = -1$ and $\cost(s',s) = 0$.
    We write $V'$ for the set of states of the extension.
\end{definition}

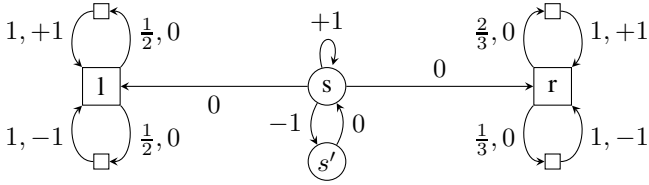
\begin{figure}[h]
    \centering
\begin{tikzpicture}[node distance=1cm and 3cm, on grid]
    \tikzstyle{acstate}=[pstate,minimum size=0.2cm,]
    \node[cstate] (Y) {s};
    \node[pstate, right=of Y] (Z) {r};
    \node[pstate, left=of Y] (X) {l};
    \node[acstate, above=of X] (lla) {};
    \node[acstate, above=of Z] (rra) {};
    \node[acstate, below=of X] (llb) {};
    \node[acstate, below=of Z] (rrb) {};
    \draw[->] (Y) edge node {$0$} (Z);
    \draw[->] (Y) edge node {$0$} (X);
    \draw[->,in=west, out=south west] (Z) edge node[left]{${\frac{1}{3}},0$} (rrb);
    \draw[->,out=east, in=south east] (rrb) edge node[right] {$1,-1$}(Z);
    \draw[->,out=north west, in=west] (Z) edge node[left]{${\frac{2}{3}},0$} (rra);
    \draw[->,in=north east, out=east] (rra) edge node[right]{$1,+1$} (Z);
    \draw[->,in=east, out=north east] (X) edge node[right]{${\frac{1}{2}},0$} (lla);
    \draw[->,out=west, in=north west] (lla) edge node[left]{$1,+1$} (X);
    \draw[->,in=east, out=south east] (X) edge node[right] {${\frac{1}{2}},0$} (llb);
    \draw[->,out=west, in=south west] (llb) edge node[left] {$1,-1$} (X);
    \node[cstate, below=of Y] (SE) {$s'$};
    \draw[->,bend right] (Y) edge node[swap]{$-1$} (SE);
    \draw[->,bend right] (SE) edge node[swap]{$0$} (Y);
    \draw[->,loop above] (Y) edge node[] {$+1$} (Y);

\end{tikzpicture}
    \caption{The extension of the system from \cref{ex:lval}.
    State $s$ has colour $1$, all others $0$.}
    \label{fig:extension}
\end{figure}

Note that the extension only incurs an $\mathcal O(\card{V_C})$ blow-up, and
$\sys{M'}$ satisfies the $\EN{k}\cap\Parity{}$ objective iff $\sys{M'}$ does.
\begin{theorem}\label{thm:extension}
    Let $\sys{M}$ be a MDP with extension $\sys{M'}$,
    $p$ be a state and $k\in\N$. Then,
    $\Val[\sys{M}]{p}{\EN{k}\cap\Parity{}} = 1$ if, and only if,
    $\Prob[\sys{M'}(\sigma)]{p}(\EN{k}\cap\Parity{}) = 1$
    for some strategy $\sigma$.
\end{theorem}
In the remainder of this section we prove this claim.
For brevity, let us write $\effect{w}$ for the cumulative cost 
$\sum_{i=1}^{k-1}\cost(s_i,s_{i+1})\in\Z$ of all steps in a finite path $w=s_1s_2\dots s_k\in V^*$.

\begin{lemma}\label{lem:extension:onlyif}
    Let $\sys{M}$ be a MDP with extension $\sys{M'}$,
    $p$ be a state of $\sys{M}$, $k\in\N$ and 
    $\sigma'$ a strategy for $\sys{M'}$
    such that $\Prob[\sys{M'}(\sigma')]{p}(\EN{k}\cap\Parity{}) = 1$.
    Then $\Val[\sys{M}]{p}{\EN{k}\cap\Parity{}} = 1$.
\end{lemma}
\begin{proof}
    Recall \cref{lem:value-one-states}, that states with $\LValP[\sys{M}]{s}=1$
    have the property that, for every $\eps>0$,
    there exists $n_{s,\eps}\in\N$ and a strategy $\sigma_{s,\eps}$ such that
    $$\Prob[\sys{M}(\sigma_{s,\eps})]{s}(\EN{n_{s,\eps}}\cap\Parity{}) \ge 1-\eps.$$
    Consider now a fixed $\eps>0$ and let $n_\eps\eqdef\max\{n_{s,\eps}\mid \LValP{s}=1\}$.
    We show the existence of a strategy $\sigma$ for $\sys{M}$ that satisfies
    ${\Prob[\sys{M}(\sigma)]{p}(\EN{k}\cap\Parity{})} \ge 1-\eps$.

    \medskip
    We propose the strategy $\sigma$ which proceeds in $\sys{M}$ just as $\sigma'$
    does in $\sys{M'}$ but skips over ``buying'' loops $(s,s')$ followed by $(s',s)$ in $\sys{M'}$.
    This goes on indefinitely unless the observed path $\rho=s_0s_1\ldots s_l$
    reaches a \emph{tipping point}: the last state $s_l$ has $\LValP[\sys{M}]{s_l}=1$ and
    the accumulated cost is $\effect{\rho}\ge n_{\eps}$.
    At this point $\sigma$ continues as $\sigma_{s_l,\eps}$.

    \medskip
    We claim that $\Prob[\sys{M}(\sigma)]{p}(\EN{k}\cap\Parity{}) \ge 1-\eps$.
    Indeed, first notice that for any prefix $\pi\in V^*$ of a run $\rho\in \Runs[\sys{M}(\sigma)]{p}$
    until the tipping point, 
    there is a unique corresponding path $\pi'=s'_1s'_2\dots s'_i\in V'^*$ in $\sys{M'}$,
    which is a prefix of some run $\rho'\in\Runs[\sys{M'}(\sigma')]{p}$.
    Moreover, the strategy $\sigma$ maintains the invariant that the accumulated cost
    of such prefix $\pi$ is
    $$\effect{\pi} = \effect{\pi'} + \card{\{j\mid s'_j\in V'\setminus V\}},$$
    the accumulated cost of the corresponding path $\pi'$
    plus the number of times $\pi'$ visited a new state in $V'\setminus V$.
    In particular this means that the path $\pi$ can only violate the energy condition
    if also $\pi'$ does.
    
    To show the claim, first notice that the error introduced by the runs in $\Runs[\sys{M}(\sigma)]{p}$
    that eventually reach a
    tipping point cannot exceed $\eps$. This is because from the tipping point onwards,
    $\sigma$ proceeds as some $\sigma_{s,\eps}$
    and thus achieves the energy-parity condition with chance $\ge 1-\eps$.
    So the error introduced by the runs in $\Runs[\sys{M}(\sigma)]{p}$
    is a weighted average of values $\le \eps$, and thus itself at most $\eps$.

    Now suppose a run $\rho\in\Runs[\sys{M}(\sigma)]{p}$ never reaches a tipping point.
    Then the corresponding run $\rho'\in \Runs[\sys{M}'(\sigma')]{p}$
    cannot visit new states in $V'\setminus V$ more than $n_\eps$ times.
    Since with chance $1$, $\rho'$
    and therefore also $\rho$ satisfies the $k$-energy condition
    it remains to show that $\rho$ also satisfies the parity condition.
    To see this, just notice that $\rho'$ satisfies this condition almost-surely
    and since it visits new states only finitely often,
    $\rho$ and $\rho'$ share an infinite suffix.
    \qedhere
\end{proof}

The ``only if'' direction of \cref{thm:extension} is slightly more complicated.
We go via an intermediate finite system $\sys{B}_k$ defined below.
The idea is that if $\EN{k}\cap\Parity{}$ holds limit-surely in $\sys{M}$
then $\Parity{}$ holds limit-surely in $\sys{B}_k$
and since $\sys{B}_k$ is finite this means that $\Parity{}$ also holds almost-surely in $\sys{B}_k$.
Based on an optimal strategy in $\sys{B}_k$
we then derive a strategy in the extension $\sys{M'}$ which satisfies $\EN{k}\cap\Parity{}$ a.s.
The two steps of the argument are shown individually as \cref{lem:M-to-B,lem:B-to-M'}.
Together with \cref{lem:extension:onlyif} these complete the proof of \cref{thm:extension}.

\begin{definition}
    Let $\sys{B}_k$ be the finite MDP that
    mimics $\sys{M}$ but
    hardcodes the accumulated costs as long as they remain between $-k$ and $\card{V}$.
    That is, the states of $\sys{B}_k$ are pairs $(s,n)$ where $s\in V$ and $-k\le n \le \card{V}$.
    Moreover, a state $(s,n)$ 
    \begin{itemize}
        \item is a (losing) sink with maximal odd parity if $n=-k$ or ${\LValP[\sys{M}]{s}}<1$,
    \item is a (winning) sink with parity $0$ if $n=\card{V}$.
    \end{itemize}
\end{definition}

We reuse strategies for $\sys{M}$ in $\sys{B}_k$
and write $\sys{B}_k(\sigma)$ for the Markov chain that is the result of basing decisions on $\sigma$
until a sink is reached.
\begin{lemma}
    \label{lem:M-to-B}
    If $\Val[\sys{M}]{s}{\EN{k}\cap\Parity{}}=1$ then $\Val[\sys{B}_k]{(s,0)}{\Parity{}}=1$.
\end{lemma}
\begin{proof}
    We show that, for every $\eps>0$, there is a strategy $\sigma$ such that
    $\Prob[\sys{B}_k(\sigma)]{s}(\Parity{})\ge 1-\eps$.
    This would be trivial (by re-using strategies from $\sys{M}$)
    if not for the extra sinks for states with
    $\LValP[\sys{M}]{s}< 1$.
    Let's call these states \emph{small} here
    and let $S$ be the set of all small states.
    We aim to show that the $k$-energy-parity condition can be satisfied
    and at the same time, the chance of visiting a small state with
    accumulated cost below $\card{V}$ can be made arbitrary small.
    More precisely, define $D\subseteq~\Runs[\sys{M}]{}$
    as the set of runs which never visit a small state with accumulated cost below $\card{V}$:
    $$D\eqdef\{s_0s_1\dots \mid \forall i\in\N.~s_i\in S \implies
    \effect{s_0\dots s_i}\ge\card{V}\}.
    $$
    We claim that
    \begin{align}
        \label{eq:M-to-B}
        &&\Val[\sys{M}]{s}{\EN{k}\cap\Parity{}\cap D} = 1
    \end{align}
    holds. We show this by contradicting the converse that, for
    $$\gamma\eqdef\Val[\sys{M}]{s}{\overline{\EN{k}\cap\Parity{}\cap D}} = 
    \Val[\sys{M}]{s}{\overline{\EN{k}\cap\Parity{}}\cup \overline{D}},$$
    $\gamma >0$.
    Equivalently, we contradict that, for every strategy $\sigma$,
    \begin{align}
        \label{eq:M-to-B'}
    \Prob[\sys{M}(\sigma)]{s}(\overline{\EN{k}\cap\Parity{}})<\gamma/2
    &\;\,\Rightarrow\;\,
    \Prob[\sys{M}(\sigma)]{s}(\overline{D})>\gamma/2.
    \end{align}
    To do this, we define $\delta<1$ as the maximum of
    $$
    \{~\Val[\sys{M}]{s}{\EN{n}\cap\Parity{}} < 1
        ~\mid~
        s\in S,~n\le k+\card{V}~
    \}\cup\{0\},
    $$
    that is,
    the maximal value $\Val[\sys{M}]{s}{\EN{n}\cap\Parity{}} < 1$
    for any $s\in S$ and $n\le k+\card{V}$,
    and $0$ if no such value exists.
    Notice that this is well defined due to the finiteness of $V$.
    This value $\delta$ estimates the chance that a run that is not in $D$ fails the $k$-energy-parity condition.
    In other words, for any strategy $\sigma$ and value $0\le \beta\le 1$,
    $$\Prob[\sys{M}(\sigma)]{s}(\overline{D})>\beta
    \text{  implies  }
    \Prob[\sys{M}(\sigma)]{s}(\overline{\EN{k}\cap\Parity{}})\ge \beta\cdot (1-\delta).$$
    This is because $\Prob[\sys{M}(\sigma)]{s}(\overline{D})$
    is the chance of a run reaching a state $s$ with accumulated cost $n<\card{V}$
    and because $\Val[\sys{M}(\sigma)]{s}{\EN{n}\cap\Parity{}} \le \delta$.

    We pick an $\eps'>0$ that is smaller than $(\gamma/2)\cdot (1-\delta)$.
    By assumption of the lemma, there is some strategy $\sigma$ such that 
    $\Prob[\sys{M}(\sigma)]{s}(\overline{\EN{k}\cap\Parity{}})<\eps'<\gamma/2$.
    Then by \cref{eq:M-to-B'}, we get
    $\Prob[\sys{M}(\sigma)]{s}(\overline{D})>\gamma/2$ and thus
    $\Prob[\sys{M}(\sigma)]{s}(\overline{\EN{k}\cap\Parity{}})\ge (\gamma/2)\cdot (1-\delta)>\eps'$,
    which is a contradiction. We conclude that \cref{eq:M-to-B} holds.

    To get the conclusion of the lemma just observe that for any strategy $\sigma$ it holds that
    $\Prob[\sys{M}(\sigma)]{s}(\EN{k}\cap\Parity{}\cap D) ~\le~ \Prob[\sys{B}_k(\sigma)]{s}(\Parity{}).$
\qedhere
\end{proof}

\begin{lemma}
    \label{lem:B-to-M'}
    If $\Val[\sys{B}_k]{s}{\Parity{}}=1$
    then
    $\Prob[\sys{M'}(\sigma')]{s}(\EN{k}\cap\Parity{}) =1$
    for some $\sigma'$.
\end{lemma}
\begin{proof}
    Finite MDPs have pure optimal strategies for the $\Parity{}$ objective \cite{Chatterjee:2004:QSP:982792.982808}. Thus
    by assumption and because $\sys{B}_k$ is finite, we can pick an optimal strategy $\sigma$ 
    satisfying $\Prob[\sys{B}_k(\sigma)]{s}(\Parity{}) =1$.
    Notice that all runs in $\Runs[\sys{B}_k(\sigma)]{s}$ according to this optimal strategy
    must never see a small state (one with $\LValP[\sys{M}]{p}<1$).
    Based on $\sigma$, we construct the strategy $\sigma'$ for $\sys{M'}$ as follows.
    
    \smallskip
    The new strategy just mimics $\sigma$ until the observed path $s_1s_2\ldots s_n$
    visits the first controlled state after a cycle with positive cost:
    it holds that $s_n\in\VC$ and there are $i,j\le n$ with $s_i=s_j$ and $\effect{s_i\ldots s_j}>0$. 
    When this happens, $\sigma'$ uses the new edges to visit a $0$-parity state,
    forgets about the cycle
    and continues just as from $s_1s_2\ldots s_is_{j+1}\dots s_n$.

    \smallskip
    We claim that $\Prob[\sys{M'}(\sigma')]{s}(\EN{k}\cap\Parity{}) =1$.
    To see this, just observe that a run of $\sys{M'}(\sigma')$ that infinitely often
    uses new states in $V'\setminus V$ must satisfy the $\Parity{}$ objective
    as those states have parity $0$.
    Those runs which visit new states only finitely often
    have a suffix that directly corresponds to a run of $\sys{B}_k(\sigma)$, and therefore also satisfy
    the parity objective.
    Finally, almost all runs in $\Runs[\sys{M'}(\sigma')]{s}$ satisfy the $\EN{k}$ objective
    because all runs in $\Runs[\sys{B}_k(\sigma)]{s}$ do, and a negative cost due to 
    visiting a new state in $V'\setminus V$ is always balanced by the strictly positive cost
    of a cycle.
    \qedhere
\end{proof}

This concludes the proof of \cref{thm:extension}.
The proof of \cref{thm:ls-energy-parity} now follows by
\cref{thm:LVAL-reach,thm:extension} and the fact that almost-sure reachability,
positive mean-payoff and $k$-energy-parity and $k$-storage-parity objectives are (pseudo) polynomial time computable (\cref{thm:as-energy-parity,thm:as-storage-parity}).

\begin{proof}[Proof of \cref{thm:main}]
Fix a MDP $\sys{M}\eqdef(\VC,\VP, E, \prob)$ with cost and parity functions.
For (1) and (2) we can, by \cref{lem:checking-LVAL-1} compute the set of control states $p$
with limit value $\LValP[\sys{M}]{p} = 1$.
Based on this, we can (in logarithmic space) construct the extension $\sys{M'}\eqdef(\VC',\VP', E', \prob')$
where
$\card{\VC'} = 2\cdot\card{\VC}$, $\card{E'} = \card{E} + 2\cdot\card{\VC}$ and the rest is as in $\sys{M}$.
By \cref{thm:extension},
a state $p\in \VC\cup\VP$ satisfies the $k$-energy-parity objective limit-surely in $\sys{M}$
iff it satisfies it almost-surely in $\sys{M}'$.
The claim then follows from \cref{thm:as-energy-parity}.

(3) To see that there are finite memory $\eps$-optimal strategies
we observe that the strategies we have constructed
in \cref{lem:extension:onlyif} work in phases
and in each phase follow some finite memory strategy.
In the first state, these strategies follow some almost-surely optimal strategy in the extension $\sys{M}'$,
but only as long as the energy level remains below some threshold that depends on $\eps$.
If this level is exceeded it means that a ``tipping point'' is reached and the strategy switches to a second phase.

The second phase starts from a state with limit value $1$, and our strategy just
tries to reach a control state in the set $A'\cup B$ from \cref{sec:limit}.
For almost-sure reachability, memoryless deterministic strategies suffice.
Finally, when ending up in a state of $A'$, the strategy follows an almost-sure optimal strategy for storage-parity
(with finite memory by \cref{thm:storage}).
Similarly, when ending up in a state of $B$, the strategy follows almost-sure optimal strategy for
the combined positive mean-payoff-parity objective (with finite memory by \cite{CD2011}).
\qedhere
\end{proof}

\section{Lower Bounds}
\label{sec:complexity}

Polynomial time hardness of all our problems 
follows, e.g., by reduction from \textsc{Reachability in AND-OR graphs} \cite{immerman1981number}, where 
non-target leaf nodes are energy decreasing sinks. 
This works even if the energy deltas are encoded in unary.

If we allow binary encoded energy deltas, i.e.~$W \gg 1$, then solving
two-player energy games is
logspace equivalent to solving two-player mean-payoff games (\hspace*{-0.4em}\cite{BFLMS2008}, Prop.~12),
a well-studied problem in $\NP\cap\coNP$ that is not known to be polynomial \cite{zwick1996complexity}.
Two-player energy games reduce directly to both almost-sure and limit-sure energy objectives for MDPs,
where adversarial states are replaced by (uniformly distributed) probabilistic ones:
a player max strategy that avoids ruin in the game directly provides a strategy
for the controller in the MDP, which means that the energy objective holds
almost-surely (thus also limit-surely).
Conversely, a winning strategy for the opponent ensures ruin after a fixed number $r$ of game rounds.
Therefore the error introduced by any controller strategy in the MDP is at least $(1/d)^r$, where $d$ is
the maximal out-degree of the probabilistic states, which means that
the energy objective cannot be satisfied even limit-surely (thus not almost-surely).
It follows that almost-sure and limit-sure energy objectives for MDPs are at least
as hard as mean-payoff games. The same holds for almost-sure and limit-sure
storage objectives for MDPs, since in the absence of parity conditions, 
storage objectives coincide with energy objectives.
Finally we obtain that all the more general 
almost-sure and limit-sure energy-parity and storage-parity objectives for
MDPs are at least as hard as mean-payoff games.

\section{Conclusions and Future Work}
\label{sec:conclusion}
We have shown that even though strategies for almost-sure energy parity objectives in MDPs
require infinite memory, the problem is still in $\NP\cap\coNP$.
Moreover, we have shown that the limit-sure problem (i.e.\ the problem of
checking whether a given configuration (state and energy level) in
energy-parity MDPs has value $1$) is also in $\NP\cap\coNP$.
However, the fact that a state has value $1$ can always be witnessed
by a family of strategies attaining values $1-\epsilon$ (for every $\epsilon >0$)
 where each member of this family uses only finite memory.

We leave open the decidability status of quantitative questions, 
e.g.\ whether $\Val[\sys{M}]{p}{\EN{k}\cap\Parity{}} \ge 0.5$ holds. 

Energy-parity objectives on finite MDPs correspond to parity objectives
on certain types of infinite MDPs where the current energy value is part of
the state. More exactly, these infinite MDPs can be described by
single-sided vector addition systems with states \cite{AMSS:CONCUR2013,ACMSS:FOSSACS2016}, where the 
probabilistic transitions cannot change the counter values but only the
control-states (thus yielding an upward-closed winning set).
I.e.\ single-sidedness corresponds to energy objectives.
For those systems, almost-sure B\"uchi objectives are decidable (even for multiple energy dimensions) \cite{ACMSS:FOSSACS2016}, but the decidability of the limit-sure problem was left open.
This problem is solved here, even for parity objectives,
but only for dimension one.
However, decidability for multiple energy dimensions remains open.

If one considers the more general case of MDPs induced by counter \emph{machines}, i.e.\ with zero-testing transitions,
then even for single-sided systems as described above all problems become undecidable from dimension 2 onwards.
However, decidability of almost-sure and limit-sure parity conditions for MDPs induced by
one-counter machines (with only one dimension of energy) remains open.

\medskip
\noindent {\bf Acknowledgements.\ }
This work was partially supported by the EPSRC through grants EP/M027287/1 and EP/M027651/1 (Energy Efficient Control), and EP/P020909/1 (Solving Parity Games in Theory and Practice).

\bibliographystyle{IEEEtran}

\bibliography{bibliography}

\newpage
\section*{Appendix}
\setcounter{section}{0}
\Alph{section}
\section{Missing Proof from Section \ref{sec:limit}}
\label{sec:appendix}
\LemMPge*
\begin{proof}
If the expected mean-payoff value of $C$ is positive, then we can assume w.l.o.g.~a pure memoryless strategy $\sigma$ that achieves this value for all states in $C$.
This is because finite MDPs allow pure and memoryless optimal strategies for the mean-payoff objective (see e.g.~\cite{LL:SIAM1969}, Thm.~1).
This strategy does not necessarily satisfy the parity objective.
However, we can mix it with a pure memoryless reachability strategy $\rho$ that moves to a fixed state $p$ with the minimal even parity $2i$ among the states in $C$.
Broadly speaking, if we follow the optimal (w.r.t.\ the mean-payoff) strategy most of the time and ``move to $p$'' only sparsely, the mean-payoff of such combined strategy would be affected only slightly. 
This can be done by using memory to ensure that the mean-payoff value remains positive (resulting in a pure finite memory strategy), or it can be done by always following $\rho$ with a tiny likelihood $\varepsilon >0$, while following $\sigma$ with a likelihood of $1-\varepsilon$ (resulting in a randomised memoryless strategy).

For the pure finite memory strategy, we can simply follow $\rho$ for $|C|$ steps (or until $p$ is reached, whatever happens earlier) followed by $n$ steps of following $\sigma$.
When $n$ goes to infinity, the expected mean payoff converges to the mean payoff of
$\sigma$.
Since the mean-payoff of $\sigma$ is strictly positive, the combined strategy
achieves a strictly positive mean-payoff already for some fixed finite $n$, and
thus finite memory suffices.

Note that using either of the just defined strategies would result in a
finite-state Markov chain with integer costs on the transitions. We can simulate such a model using {\em probabilistic one-counter automata} \cite{etessami2010quasi}, where the energy level is allowed to change by at most $1$ in each step, just by modelling an increase of $k$ by $k$ increases of one.
Now we can use a result by Br{\'{a}}zdil, Kiefer, and Ku\v{c}era \cite{BKK/14}
for such a model for the case where it consists of a single SCC 
(which is the case here, because of the way $\sigma$ is defined). 
In particular, Lemma~5.13 in \cite{BKK/14} established an upper bound on the
probability of termination (i.e.\ reaching energy level $0$) in a
probabilistic one-counter automaton with a positive mean-payoff (referred to
as `trend' there) where starting with energy level $k$. This upper bound can
be explicitly computed for any given probabilistic one-counter automaton and
energy level $k$. However, for our purposes, it suffices to note that this
bound converges to $0$ as $k$ increases. This shows that the probability of
winning can be made arbitrarily close to $1$ by choosing a sufficiently high
initial energy level and using the strategy defined in the previous paragraph.
Thus the states in $C$ indeed have limit value $1$.
\qedhere
\end{proof}

\section{Memory Requirements for $\eps$-optimal Strategies}
\label{app:finite}
In this appendix, we discuss the complexity of the strategies needed. First, we show that the strategies for determining the limit values of states are quite simple: they can either be chosen to be finite memory and pure, or randomised and memoryless.
For winning limit-surely from a state energy pair, finite memory pure strategies suffice, but not necessarily memoryless ones, not even if we allow for randomisation.

We start by showing the negative results on examples.
\Cref{fig:memory} shows an MDP, where it is quite easy to see that both states have limit value $1$.
However, when looking at the two memoryless pure strategies, it is equally clear that
either (if the choice is to move from $s$ to $r$) the energy condition is violated almost-surely,
or (if the choice is to remain in $s$), the parity condition is violated on the only run.
Nevertheless, the state $s$ satisfies the $k$-energy-parity objective limit-surely, but
not almost-surely, for any fixed initial energy level $k$.

\begin{figure}[h]
  \begin{minipage}[t]{.48\textwidth}
  \centering
  \begin{tikzpicture}[node distance=1cm and 3cm, on grid]
    \tikzstyle{acstate}=[cstate,minimum size=0.2cm]
    \node[cstate] (Y) {s};
    \node[pstate, right=of Y] (Z) {r};
    \draw[->] (Y) to [bend left] node[midway] {$-1$} (Z);
    \draw[->] (Z) to [bend left] node[midway] {$\frac{1}{2}, -1$} (Y);
    \draw[->,loop above] (Z) edge node[] {$\frac{1}{2}, -1$} (Z);
    \draw[->,loop above] (Y) edge node[] {$+1$} (Y);
\end{tikzpicture}
\caption{
\label{fig:memory}
An energy-parity MDP that requires memory or randomisation to win with a
positive probability, from any initial energy level. States $s$ and $r$ have
priorities $1$ and $0$, respectively.
} 
    \end{minipage}
    \quad
  \begin{minipage}[t]{.48\textwidth}
  \centering
   \begin{tikzpicture}[node distance=1cm and 2cm, on grid]
     \tikzstyle{acstate}=[cstate,minimum size=0.2cm]
     \node[cstate] (X) {p};
     \node[cstate, right=of X] (Y) {s};
     \node[pstate, right=of Y] (Z) {r};
     \draw[->] (Y) to [bend left] node[midway] {$+1$} (X);
     \draw[->] (X) to [bend left] node[midway] {$+1$} (Y);
     \draw[->] (Y) to [bend left] node[midway] {$-1$} (Z);
     \draw[->] (Z) to [bend left] node[midway] {$\frac{1}{4}, -1$} (Y);
     \draw[->,loop above] (Z) edge node[] {$\frac{3}{4}, +1$} (Z);
 \end{tikzpicture}
 \caption{
     \label{fig:3states}
      Energy-parity MDP where a randomised memoryless strategy does not suffice for limit-sure winning for any initial energy level. States $s$ and $r$ have priority $2$ and state $p$ has priority $1$.} 
    \end{minipage}
\end{figure}

\Cref{fig:3states} shows an energy-parity MDP, where all states have limit value $1$, and the two left states have limit value $1$ even from zero energy. (They can simply boost their energy level long enough.)
Only in the middle state do we need to make choices. For all memoryless randomised strategies that move to the left with a probability $>0$, the minimal priority on all runs is almost-surely $1$, such that these strategies are almost-surely losing from all states and energy levels. The only remaining candidate strategy is to always move to the right. But, for all energy levels and starting states, there is a positive probability that the energy objective is violated. (E.g.\ when starting with energy $k$ in the middle state, it will violate the energy condition in $k+1$ steps with a chance $4^{-\lceil k/2\rceil}$.)

To see that finite memory always suffices, we can simply note that the strategies we have constructed work in stages.
The `energy boost' part from \cref{sec:extensions} does not require memory on the extended arena (and thus finite memory on the original arena). Further memory can be used to determine when there is sufficient energy to progress to the strategy from \cref{sec:limit}.

The strategy for \cref{sec:limit} consists of reaching $A'$ or a positive $2i$ maximal set almost-surely and then winning limit-surely there.
For almost-sure reachability, memoryless deterministic strategies suffice. The same holds for winning in $A'$.
For winning in a positive $2i$ maximal set, the proof of \cref{lem:MPge0} also establishes that pure finite memory and randomised memoryless strategies suffice.

\end{document}